\documentclass[10pt,journal,twocolumn]{IEEEtran}%,draftclsnofoot
\usepackage{amsmath}%
\usepackage{times}
\usepackage{textcomp}
\usepackage{amsfonts}%
\usepackage{amssymb}%
\usepackage{amsthm}
\usepackage{graphicx}%
\usepackage{array}%
\usepackage{subfig}%
\usepackage{color}%
\usepackage{soul}
\newtheorem{theorem}{Theorem}

\newtheorem{assumption}{Assumption}

\newtheorem{corollary}{Corollary}

\newtheorem{definition}{Definition}

\newtheorem{lemma}{Lemma}

\newtheorem{remark}{Remark}

\newcommand{\cM}{\mathcal{M}} % highly popular items
\newcommand{\xs}{x^*} % optimal proactive download

\newcommand{\vx}{{\bf x}}
\newcommand{\vxb}{\bar{\bf x}}
\newcommand{\E}{\mathbb{E}}

\newcommand{\vxs}{{\bf x}^*}

\newcommand{\vps}{{\bf p}^*}

\newcommand{\cB}{{\mathcal{B}}}
\newcommand{\cP}{{\mathcal{P}}}
\newcommand{\cJ}{{\mathcal{J}}}

\newcommand{\cN}{{\mathcal{N}}}

\newcommand{\cF}{\mathcal{F}}

\newcommand{\vv}{{\bf v}}
\newcommand{\vr}{{\bf r}}
\newcommand{\vnm}{v_n(m)}
\newcommand{\rnm}{r_n(m)}

\newcommand{\pnt}{P_{n,t}}

\newcommand{\tpnt}{\tilde{P}_{n,t}}
\newcommand{\pntb}{\bar{P}_{n,t}}
\newcommand{\qnt}{\tilde{q}_{n,t}}
\newcommand{\vpnt}{{\bf p}_{n,t}}
\newcommand{\vtpnt}{\tilde{\bf p}_{n,t}}
\newcommand{\vp}{{\bf p}}
\newcommand{\vtp}{\tilde{\bf p}}
\newcommand{\vpb}{\bar{\bf p}}
\newcommand{\Smin}{\check{S}}
\newcommand{\Smax}{\hat{S}}
\newcommand{\Int}{\mathbb{I}_{n,t}}

\newcommand{\Intb}{\bar{\mathbb{I}}_{n,t}}

\newcommand{\I}{\mathbb{I}}
\newcommand{\tpint}{\tilde{\pi}_{n,t}}
\newcommand{\vtpint}{\tilde{\boldsymbol{\pi}}_{n,t}}

\newcommand{\an}{\alpha_n}

\newcommand{\vhp}{\hat{\bf p}}
\newcommand{\vhx}{\hat{\bf x}}
\newcommand{\vtx}{\tilde{\bf x}}
\newcommand{\fh}{\hat{f}}
\newcommand{\vhpnt}{\hat{\bf p}_{n,t}}

\newcommand{\hpnt}{\hat{P}_{n,t}}
\newcommand{\vtpntz}{\tilde{\bf p}_{n_0,t_0}}
\newcommand{\vbpntz}{\bar{\bf p}_{n_0,t_0}}

\newcommand{\vbpnt}{\bar{\bf p}_{n,t}}

\newcommand{\txnt}{\tilde{x}_{n,t}}
\newcommand{\hxnt}{\hat{x}_{n,t}}

\newcommand{\pp}{p_1}
\newcommand{\po}{p_0}
\newcommand{\nz}{n_0}
\newcommand{\tz}{t_0}
\newcommand{\gmax}{\hat{\gamma}}
\newcommand{\gmin}{\check{\gamma}}
\newcommand{\cBmax}{\hat{\cB}}
\newcommand{\cBmin}{\check{\cB}}
\newcommand{\vnt}{v_{n,t}}
\newcommand{\Cb}{\bar{c}}
\newcommand{\Co}{\overset{\circ}{c}}

\providecommand{\norm}[1]{\lVert#1\rVert}
\begin{document}
\title{Proactive Content Download and User Demand Shaping for Data Networks}

\author{
%\authorblockN{Names}
%\authorblockA{Wireless Intelligent Networks\\
%Center (WINC)\\
%Nile University, Cairo, Egypt\\
%....@nileu.edu.eg}
%}
\IEEEauthorblockN{John Tadrous, Atilla Eryilmaz, and Hesham El Gamal}
\thanks{Authors are with the Department of Electrical and Computer Engineering at the Ohio State University, Columbus, USA. E-mail: \{tadrousj,eryilmaz,helgamal\}@ece.osu.edu. \newline This work is primarily supported by the QNRF grants NPRP 7-923-2-344 and NPRP 09-1168-2-455. Also, the work of A. Eryilmaz was in part supported by the NSF grant CAREER-CNS-0953515. }}
\maketitle

\begin{abstract}

In this work, we propose and study optimal proactive resource allocation and demand shaping for data networks. %It has been recently reported that service providers (SPs) incur excessive costs to sustain reliable data delivery for the customer base. On the other hand, there is a growing body of evidence that the available resources, especially for wireless networks, suffer a substantial underutilization due to a large disparity between the peak and average demand patterns.
 Motivated by the recent findings on the predictability of human behavior patterns in data networks, and the emergence of highly capable handheld devices, our design aims  to smooth out the network traffic over time and minimize the data delivery costs.

Our framework utilizes proactive data services as well as smart content recommendation schemes for shaping the demand. Proactive data services take place during the off-peak hours based on a statistical prediction of a \emph{demand profile} for each user, whereas smart content recommendation assigns modified valuations to data items so as to render the users' demand less uncertain. Hence, our recommendation scheme aims to boost the performance of proactive services within the allowed flexibility of user requirements. We conduct theoretical performance analysis that quantifies the leveraged cost reduction through the proposed framework. We show that the \emph{cost reduction} scales at the same rate as the cost function scales with the number of users. Further, we prove that \emph{demand shaping} through smart recommendation strictly reduces the incurred cost even below that of proactive downloads without recommendation.
\end{abstract}


\section{Introduction}
\label{sec:intro}
\IEEEPARstart{T} he vast expansion of highly capable smart wireless devices has powered a substantial transformation of the global data network into a more mobile, increasingly demanding, and socially more interconnected form. Such a wireless revolution has raised major concerns about an inevitable surge of wireless traffic load by throughput-hungry applications that existing resource allocation schemes may not stand. These applications include multimedia services which constitute more than $50\%$ of the total Internet load \cite{cisco1}. On the other hand, there is a growing body of evidence that the available spectrum, which defines an ultimate resource for wireless communications, is suffering an inherent underutilization problem as has been reported in the recent studies by FCC.

 Consequently, there is an urgent call for the development of more advanced and sophisticated techniques improving the wireless resource management and allocation. The notion of dynamic spectrum access (DSA) has been introduced as a remedy to the spectrum underutilization problem, and has been enabled through the cognitive radio technology \cite{Mitola}. The cognitive radio approach, however, is still facing significant technological hurdles \cite{Gridlock}, and will offer only a partial solution to the problem. This limitation is tied to the main reason behind the underutilization of the spectrum; namely \emph{the large disparity between the average and peak traffic demand in the network}.

The traffic demand at the peak hour is much higher than that at night. The cognitive radio approach assumes that the users will be able to utilize the spectrum at the off-peak times, but at those times one may
expect the cognitive radio traffic characteristics to be similar to that of the primary users.% (e.g., at night most of the primary and cognitive radio users are expected to be idle). 

 Recently, we have proposed the notion of \emph{proactive resource allocation} \cite{TEG12} as a remedy to the spectrum crisis. The technique aims at exploiting the \emph{predictable} human demand as well as the powerful processing capabilities and the large memory storage offered by the smart wireless devices in smoothing out the wireless network traffic over time by \emph{proactively serving  predictable peak-hour requests during the off-peak time}. Hence, the peak-to-average demand ratio is minimized and significant utilization for the available resources is provided.

There exists a substantial evidence, both for general human behavior \cite{EP06}-\cite{SQBB10} and specifically for wireless data users \cite{RRDTool}-\cite{CF6}, that supports the underlying premise of large-timescale (in the order of minutes to hours) user predictability which, in turn, motivates our proactive design framework. {In \cite{TEG12}, the potential gains of proactive resource allocation, under large-timescale optimization, have been characterized for unicast and multicast networks as well as networks with heterogeneous QoS requirements. %In \cite{Allerton10}, the notion of proactive resource allocation for wireless unicast networks has been introduced and analyzed under perfect predictability of users' demand, and the performance has been quantified through the diversity gain metric. The results have been extended to multicast networks in \cite{ISIT11} where multicast alignment gain revealed a potential for a significant reduction in the resources required to attain a certain level of quality of service (QoS). In \cite{Asilomar11}, proactive resource allocation has been investigated in cognitive radio networks where the \emph{good citizen} phenomenon is demonstrated. Such a phenomenon has revealed an enhanced QoS for a \emph{non-proactive} cognitive user while primary users employ proactive resource allocation. 
In this work, however, user predictability has been assumed to be perfect, and network schedulers capable of serving the demand proactively at zero cost as long as the load falls within a predefined capacity. In this paper, we shed light on the cost incurred due to data delivery, including proactively served data. Further, we consider the uncertainty about the future demand and capture it through probabilistic demand profiles.}  

{There exist attempts to leverage the popularity and/or predictability of some media content in pre-caching purposes, and hence enhance the reliability of communication. We can see in \cite{BI05} that the authors apply a proactive scheme for emptying network queues that store MPEG videos in order to accommodate for the incoming traffic with the objective of minimizing the frame error rate. Yet, the developed approach lacks solid performance analysis, and its impact is limited to the small-timescale of milliseconds. The content pre-caching idea has also been proposed in \cite{GH04} with the intention of offering high QoS communication in some geographical areas. While the authors present their practical views of the design, there are no results provided to quantify the potential of the idea. More recently, in \cite{BWZZ12}, the popularity of video content has been harnessed to provide proactive caching solutions that essentially minimize data transmission power. However, the proposed model assumes multicast transmission, hence rules out unicast sources and the associated concerns about resource losses due to uncertainty. Further, the assumed model treats all demand as identically distributed over time, which therefore does not capture the peak-to-average-ratio problem addressed in this work.}
%There are also emerging works aiming to balance the network traffic over time through \emph{time/load-dependent pricing} (see for example \cite{GL10}-\cite{WHSC11}). The main approach is to adjust the service price depending on the total network load in a way that assigns low prices to off-peak services and higher prices to peak-hour demand. A recent study in \cite{GL10} highlights the potential of smoothing-out the network traffic through time-dependent pricing. In fact, some network operators outside the USA (such as Orange, MTN and Uninor) have already started using adaptive pricing schemes to mitigate the excessive cost resulting from high bandwidth consumption \cite{E09}-\cite{H10}.

 %On the other hand, there exists a recent work \cite{WHSC11} that considers the same problem in the USA. It optimizes price allocation mechanisms that encourage network users to \emph{delay} their demand to the off-peak hour based on collected statistics about their willingness to defer the demand to a time when the service prices are considerably reduced. While pricing can be used for demand shaping (see e.g. \cite{SDP13}), our proposed approach sheds light on the use of recommendation schemes as another important parameter that controls the future demand. Further, we highlight the point that our proposed approach \emph{does not} aim at pushing users to defer their demand in time. Instead, it encourages them to be deterministic about their future demand.
 
 In this paper, we develop a framework for optimal proactive resource allocation for wireless networks comprising a service provider and associated subscribers. The subscribers' demand assumes cyclostationary statistics as it repeats itself over finite time durations, whereby the service provider can track, learn and construct a {\bf demand profile} to each subscriber. These profiles are then used to determine proactive data downloads to the users in a way that minimizes the time average expected cost incurred by the service provider while providing reliable data delivery. Moreover, we investigate a further improvement to the cost performance under the potential of slightly modified demand profiles. We refer to this operation as {\bf demand shaping}. {While such an operation is primarily targeted to further reduce the data delivery costs, it also has a substantial promise for high QoS data delivery. In particular, proactive data download facilitates an enhanced consumption process without undesirable delays or buffering problems. Thus, it implicitly provides an incentive for users cooperation with demand shaping.} We develop a smart recommendation scheme  to perform demand shaping while maintaining the user satisfaction about the quality of offered data items. %The demand shaping scheme is proved to strictly reduce the cost even below proactive downloads alone.
  The main contributions of this work are listed as follows.

$\bullet$ In Section \ref{sec:sys_mod}, we provide a description of the time-slotted system model, %present the notion of user demand profile and its \emph{cyclostationary} nature,
and layout the time average expected cost expression for a traditional \emph{non-proactive} network.

$\bullet$ In Section \ref{sec:prob_form}, we formulate the problem, and prove the existence of a steady state solution.

$\bullet$ In Section \ref{sec:pro_down}, we consider the proactive data download side of the problem. We adopt the cost reduction leveraged through proactive downloads as our metric of interest. An upper and lower bounds on the optimal cost reduction are established, and its asymptotic scaling laws with the number of users are characterized.% We show that, successive time slots that witness disparate levels of traffic load result in a cost reduction that scales with the same order of the non-proactive cost itself.

$\bullet$ In Section \ref{sec:Demand_Shaping}, we study the joint allocation of proactive data downloads and user demand profiles.  In Section \ref{sec:w_constraint},  we characterize the optimal solution to the problem under weak user satisfaction constraints.% where each user is flexible enough to follow the recommendations made by the service provider.

$\bullet$ In Section \ref{sec:s_constraint}, we take the user satisfaction into account. We propose an iterative scheme where an almost-surely \emph{strictly} reduced cost below that of proactive downloads alone can be obtained.

$\bullet$ In Section \ref{sec:Recom}, we formulate and study a data-item recommendation problem that assigns new ratings to the recommended data items that 1) achieve the modified profiles, and 2) remain as close as possible, in the Euclidean- distance sense, to the actual ratings made by the corresponding users.

%$\bullet$ The work is concluded in Section \ref{sec:conc}.

\section{System Model}
\label{sec:sys_mod}
We consider a network comprising $N$ users with variable demand and a service provider that responds to user's requests in a timely basis. The service provider has a total of $M$ different data items that each user can request in a random fashion.
Each data item $m$ is assumed to have a size of $S(m)>0$, and 
\begin{eqnarray}
\label{eq:Smin}
\Smin:=&\displaystyle{\min_{m\in\{1,\cdots,M\}}}\{S(m)\}>0,\\
\label{eq:Smax}
\Smax:=&\displaystyle{\max_{m\in\{1,\cdots,M\}}}\{S(m)\}<\infty.
\end{eqnarray}

The assumption on the fixed number of data items is justified by the observation that new content updates replace obsolete data. For example, several YouTube channels, CNN, Fox news, and social network updates gain the highest attention and popularity, while previous updates are substantially less interesting. %Further, the fixed data item size is an application-layer approximation which categorizes content updates for the same data item into equal-sizes so as to capture the amount of resources consumed to sustain their intact delivery.}

In a time-slotted system, the content of each data item is consistently updated every time slot, where such a content could be a movie (as in YouTube and Netflix), a soundtrack (as in Pandora), a social network update (as in Facebook and Twitter), a news update (as in CNN and Fox News), etc. {Our assumption of fast content update rate is intended to yield the worst-case scenario results and to reveal the potential gains to be leveraged.} We consider the application-layer timescale in which the duration of a time slot is the time taken by a user to completely run the requested data item, which can be of the order of minutes or possibly hours\footnote{{Please note that a \emph{data item} is not restricted to contain only one video, as for example one can consider CNN as a data item with multiple new videos update it every time slot. A user therefore can consume all these videos in one slot, while another user consumes only one video with roughly the same duration from a YouTube channel which can be treated as another data item. }}. At each time slot, the service provider supplies the requested data items, which have been updated, {and partially served} over the previous time slot to the respective users.

\textbf{Data item valuations:} Define $\vnt(m)\in[0,1]$ as the valuation (rating) of data item $m$ as offered by the service provider to user $n$ at time $t$. The offered valuations depend on the user preferences and interests, and can be estimated through different techniques, such as collaborative filtering.

%\subsection{User's Demand Profile}
\textbf{Users' demand profiles: } We assume that the demand of each user can be tracked, learned, and predicted by the service provider over time.  The service provider constructs a \emph{demand profile} for every user $n$ and time slot $t$, denoted $\vpnt=(\pnt(m))_{m=1}^{M}$, where $\pnt(m)$ is the probability that user $n$ requests item $m$ in slot $t.$ { The probability that user $n$ requests data item $m$ at time slot $t$ is modeled as 
$\pnt(m):=\phi_{m,t}(\vv_{n,t})$, where $\phi_{m,t}:[0,1]^M\to[0,1]$ is a non-negative function that maps the ratings of user $n$ into a corresponding probability of requesting item $m$ specifically at time slot $t$ for any $m,n,t$. The function $\phi_{m,t}$ captures the relative differences between the ratings of the $M$ data items as seen by user $n$.}%We assume that \emph{each user can request at most one data item per time slot}.
 Then, we model the statistics of the predictable user demands as follows: \\
\begin{itemize}
\item The demand of user $n$ at slot $t$ is captured by a random variable $\Int(m)$ where
	\begin{equation*}
	\label{eq:I}
	\Int(m)=\begin{cases}1, & \text{with probability } \pnt(m),\\
	                     0, & \text{with probability } 1-\pnt(m).
	\end{cases}
	\end{equation*}
%	\item For any time slot $t$, $\Int(m)$ is independent of $\I_{n,t+1}(k)$ for all $m,k$.
	\item For any two users $n,k$ such that $n\neq k$, $\Int(m)$ is independent of $\I_{k,t}(j)$ for all $m,j$.
	\item At slot $t\geq 0$, user $n$ requests at most one data item. Hence $\sum_{m=1}^{M}\Int(m)\leq 1$.
	\item The probability that user $n$ does not request any data item at slot $t$ is $\qnt:=1-\sum_{m=1}^{M}\pnt(m)$.
\end{itemize}

\begin{figure}
	\centering
		\includegraphics[width=0.4\textwidth]{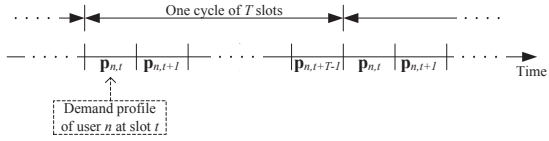}
	\caption{Cyclostationary demand profiles over time. Every $T$ time slots, the demand profile of user $n$ is repeated.}
	\label{fig:Cyclo}
\end{figure}

Further, the demand profile of each user follows a \emph{cyclostationary} pattern that repeats itself consistently in a period of $T$ time slots as shown in Fig. \ref{fig:Cyclo}. The $T$-slot period can be interpreted as a single day through which the activity of each user varies each hour, but occurs with the same statistics consistently each day. Thus, we can write $\vp_{n,t+l}=\vp_{n,t+kT+l}$ for any non-negative integer $k$, and $l=0,\cdots,T-1$. {The cyclostationarity assumption is motivated by the recent findings in human behavioral patterns \cite{EP06,FG08}, and the wireless user activity measurements provided in \cite{RRDTool}. Yet, the results and insights drawn under such an assumption can directly apply to the case when this assumption is relaxed.}

\textbf{Incurred cost: } To supply requested data items, the service provider incurs a certain cost due to the resources consumed at each time slot, which is supposed to depend on the total load created by the users' demand. Suppose that $L$ is the total network load at a given slot, then the cost incurred by the service provider is $C(L)$, where $C:\mathbb{R}_+\to\mathbb{R}_+$ is a smooth, strictly convex, and monotonically increasing cost function. We assume also that {full demand realizations of any slot are available to service provider at the end of such slot, thus proactive control decisions are made based on the statistical characteristics of demand\footnote{{While static allocation of proactive downloads may  be out performed by a dynamic policy, it will facilitate the development of simple policies that achieve optimal-order scaling, as shown in Section \ref{sec:pro_down}.}}.} 

 %However, this assumption does not contradict the assumption that service provider can construct users' profiles through tracking and learning their demand since such a process is supposed to have taken place through a long period of time over which the service provider can retrieve sufficient statistics about the users' preferences and activities from the content delivery network, such a step that we assume to be performed initially. 

%\subsection{Baseline Non-proactive Network}
We consider the time average expected cost incurred by a \emph{non-proactive} network, a network whose service provider does not exploit the predictability of the user demand, as
\begin{equation}
\label{eq:N_inf_horizon}
C^{\cN}(N):=\limsup_{\tau\to\infty}\frac{1}{\tau}\sum_{t=0}^{\tau-1}\E\left[C\left(L_{t}\right)\right],
\end{equation}
where 
%\begin{equation}
$L_t:=\sum_{m=1}^{M}\sum_{n=1}^{N}S(m)\Int(m), \quad t\geq 0$
%\end{equation}
is the total load at time slot $t$ encountered by the non-proactive network, { and the expectation is over the demand profiles $\vpnt$, $\forall n,t$.} The average cost $C^{\cN}(N)$ represents a baseline to which we compare the relative cost reduction leveraged through efficient utilization of the available users' profiles.

\section{Problem Formulation}
%In this section, we pose the general formulation of the proactive resource allocation problem, and highlight some of its features. 
\label{sec:prob_form}
\subsection{Proactive Data Download and Demand Shaping}
The proposed \emph{proactive} network framework utilizes the predictable and cyclostationary nature of the users' demand in balancing the total load over time and minimizing the time average expected cost.

To this end, our approach aims to produce proactive data downloads at each time slot depending on the demand statistics. It harnesses the prior knowledge about the user profile in the upcoming time slot, combines it with the statistics about the demand during the current slot, and proactively sends a portion of each potential data item to the respective users. We denote by $x_{n,t+1}(m)$ the portion of data item $m$ sent ahead to user $n$ at time slot $t$, $m=1,\cdots,M$, $n=1,\cdots,N$, $t=0,1,\cdots$. % Fig. \ref{fig:Time} provides an illustrating diagram of the proactive data download strategy. 
{We note that, the limitation of proactive download decisions to one-slot ahead is aligned with our assumption that content updates occur every time slot, hence small proactive download window guarantees freshness of served content.}
% \begin{figure}
%	\centering
%		\includegraphics[width=0.4\textwidth]{Time.eps}
%	\caption{Time diagram of proactive downloads for the demand of user $n$.}
%	\label{fig:Time}
% \end{figure}

We also investigate the potential for a service provider to \emph{slightly} modify the demand profile of each user {so as to} enhance the \emph{certainty} about the future requests under user \emph{satisfaction} constraints. {Note that user satisfaction is a crucial parameter that service providers need to improve, and modified demand profiles have to obey.}

 The modification of demand profiles offers a further cost reduction gain since  more deterministic users render the proactive download process more efficient\footnote{Addressed in more detail in Section \ref{sec:Demand_Shaping}.}. %In fact, the statistical knowledge about the future demand might result in potential wastage in the proactive downloads as such data may not be actually requested and hence the service provider could end up incurring extra cost. 

The service provider may offer slightly different valuations from those recognized by users as follows. Suppose that $\vtpnt=(\tpnt(m))_{m=1}^{M}$ is the profile of user $n$ at time slot $t$ with $\qnt$ being the probability that user $n$ remains silent in slot $t$. Then, we model user $n$'s flexibility to change his own profile of slot $t$ from $\vtpnt$ to $\vpnt$ by a satisfaction region $\cF_{\vtpnt}$ which is a collection of demand profiles that user $n$ is satisfied to adopt at slot $t$. Further, each profile $\vpnt\in\cF_{\vtpnt}$ has to always satisfy $\qnt=1-\sum_{m=1}^{M}\pnt(m)$, since the modification of the user preferences should not affect his activity at that slot. {The dependence of the satisfaction region on $\vtpnt$ demonstrates the importance of that profile in the new allocation so as not to deviate much from it.}

 For simplicity of notation, let $\vv=\{\vnt\}_{n,t}$, $\vx=\{\vx_{n,t}\}_{n,t}$, $\vp=\{\vpnt\}_{n,t}$, and $\Cb(.)=\limsup_{\tau\to\infty}\frac{1}{\tau}\sum_{t=0}^{\tau-1}\E[C(.)]$.

Now, the proactive download and demand shaping problem for time average expected cost minimization is formulated as $C^{\cP}(N):=$
\begin{equation}
\label{eq:P_inf_horizon}
\begin{split}
 & \underset{\vv,\vx,\vp}{\min} \quad\Cb\left(L_t+\sum_{m=1}^{M}\sum_{n=1}^{N}(x_{n,t+1}(m)-x_{n,t}(m)\Int(m))\right)\\  
%\end{aligned}
%\end{equation}
%\begin{equation*}
%\begin{aligned}
  & \text{subject to}\\
  & \quad 0\leq x_{n,t}(m) \leq S(m), \quad \forall m,n,t\geq 0,\\
  & \quad \vpnt\in\cF_{\vtpnt}, \quad \forall n,t\geq 0,\\
  & \quad \vpnt=\vp_{n,kT+t}, \quad \forall k,n,t\geq 0\\
  & \quad \pnt(m)=\phi_{m,t}(\vv_{n,t}),\quad \forall m,n,t\geq 0, \\
  & \quad \vv_{n,t}\in[0,1]^M,\quad \forall n=1,\cdots,N,
\end{split}
\end{equation} 
where the optimization is jointly done over the valuations, proactive downloads and users' profiles, {with respect to which the expectation is defined}. 
We now investigate the existence of a \emph{one-cycle} steady-state solution.

\subsection{One-cycle Steady-state Solution}
{Since the demand profiles of the users exhibit a cyclostationary nature, it is expected that the optimal proactive downloads will also be periodic over time as they only depend on the demand dynamics. On the other hand, the valuations required to attain modified profiles are essentially periodic based on the mapping $\phi_{m,t}$. This observation}
suggest the existence of a steady-state solution to \eqref{eq:P_inf_horizon} {which is proved in Theorem \ref{th:steady_state}}. 

\begin{lemma}
\label{lem:avg_conv}
Let $\{y_n\}_{n\geq 0}$ be a {periodic} sequence in $\mathbb{R}_+$, {with a finite period $T$}, then {$\lim_{n\to\infty}\frac{1}{n}\sum_{i=0}^{n-1}y_i$ $=$ $\frac{1}{T}\sum_{i=0}^{T-1}y_i$}.
\end{lemma}
%\begin{proof}
%Please refer to Appendix \ref{app:avg_conv}.
%\end{proof}

%\begin{corollary}
%For any $N\in\mathbb{N}$, the time average expected costs $C^{\cN}(N)$ and $C^{\cP}(N)$ exist, and $\limsup_{\tau\to\infty}$ can be replaced with $\lim_{\tau\to\infty}$ in \eqref{eq:N_inf_horizon}, and \eqref{eq:P_inf_horizon}.
%\end{corollary}
%\begin{proof}
%Follows directly from Lemma \ref{lem:avg_conv} while noting that $0\leq \E\left[C(L_t)\right]\leq C(\Smax\cdot N)$ and 
%\begin{equation*}
%\begin{split}
%0 & \leq \E\left[C\left(L_t+\sum_{m=1}^{M}\sum_{n=1}^{N}x_{n,t+1}(m)-x_{n,t}(m)\Int(m)\right)\right]\\
%  & \leq C(\Smax\cdot N), \quad \text{ for all }t\geq 0.
%\end{split}
%\end{equation*}
%\end{proof}

\begin{theorem}
\label{th:steady_state}
For any $N\in\mathbb{N}$, {optimization problem \eqref{eq:P_inf_horizon} is equivalent to} $C^{\cP}(N)=$%the joint proactive download and demand shaping problem can be formulated as
\begin{equation}
\label{eq:P_f_horizon}
\begin{aligned}
&  \underset{\vv,\vx,\vp}{\min}\quad\Co\left(L_t+\sum_{m=1}^{M}\sum_{n=1}^{N}x_{n,t+1}(m)-x_{n,t}(m)\Int(m)\right)\\  
\end{aligned}
\end{equation}
\begin{equation*}
\begin{aligned}
& \text{subject to}\\
& \quad 0\leq x_{n,t}(m) \leq S(m), \quad \forall m,n,t=0,\cdots,T-1,\\
& \quad \vpnt\in\cF_{\vtpnt}, \quad \forall n,t=0,\cdots,T-1,\\
& \quad \vx_{n,0}=\vx_{n,T},\\
& \quad \pnt(m)=\phi_{m,t}(\vv_{n,t}),\quad \forall m,n,t=1,\cdots,T-1, \\
& \quad \vv_{n,t}\in[0,1]^M,\quad \forall n=1,\cdots,N,
& \end{aligned}
\end{equation*}
where $\Co(.)=\frac{1}{T}\sum_{t=0}^{T-1}\E[C(.)]$.
\end{theorem}

\begin{proof}
Proof is omitted for brevity. We refer interested readers to Appendix A in \cite{TR}.
\end{proof}

\begin{corollary}
For any $N\in\mathbb{N}$, the time average expected cost incurred by the non-proactive network is given by
\begin{equation}
\label{eq:N_f_horizon}
C^{\cN}(N)=\Co(L_t).
\end{equation}
\end{corollary}
As such, only one cycle of the time average expected cost is tantamount to the infinite-horizon time average cost. {That is, over a period of $T$ slots, the service provider can assign optimal quantities (i.e., proactive downloads, valuations, and demand profiles) which can be applied to all time slots and optimize the infinite-horizon average cost.} 
It turns out, however, that the general formulation in \eqref{eq:P_f_horizon} is not convex in $(\vx,\vp)$ because the objective function involves a collection of products of the components of both $\vx$ and $\vp$, where $\vx=\{\vx_{n,t}\}_{n,t}$, $\vp=\{\vp_{n,t}\}_{n,t}$. In order to tackle the difficulty in obtaining a global optimal solution for the problem, we divide it into three main steps. We first address the performance of proactive downloads alone without demand shaping which characterize the the best proactive data allocation in response to a given demand profile. Then, we consider the determination of the best joint demand profile and proactive download that always improves beyond proactive downloads alone. Finally, we investigate data item valuation approach that yields the target demand profile.

\section{Proactive Data Download}
\label{sec:pro_down}
In this section, we quantify the performance of proactive data downloads only. {In particular,} we assume that the system is operating at a given demand profile $\vtp$ while the service provider is determining the optimal proactive downloads $\vxs$. {We show that the cost reduction leveraged through efficient proactive downloads grows unboundedly with the number of users.} {The problem is now convex in $\vx$, and by Theorem \ref{th:steady_state}, this case is again equivalent to optimizing the following:} $ C^{\cP}(N,\vtp):= $
\begin{equation}
\label{eq:CP_avg}
\begin{aligned}
& \underset{\vx}{\min}  \quad \Co\left(L_t+\sum_{m=1}^{M}\sum_{n=1}^{N}x_{n,t+1}(m)-x_{n,t}(m)\Int(m)\right)\\  
&\text{subject to}
 \quad 0\leq x_{n,t}(m) \leq S(m), \quad \forall m,n,t,\\
\end{aligned}  
\end{equation} 
and 
$C^{\cN}(N,\vtp):=\frac{1}{T}\sum_{t=0}^{T-1}\E\left[C(L_t)\right]$.
In this section, we focus on the \emph{cost reduction} leveraged through efficient proactive downloads solely. We define the cost reduction when there are $N$ users in the system as 
\begin{equation*}
\Delta C(N):=C^{\cN}(N,\vtp)-C^{\cP}(N,\vtp),
\end{equation*}
and consider its asymptotic performance when the number of users grows to infinity. More specifically, 
\begin{equation*}
\limsup_{N\to\infty}\frac{\Delta C(N)}{h(N) C'(\gamma\cdot N)}, \quad \text{and } \liminf_{N\to\infty}\frac{\Delta C(N)}{h(N) C'(\gamma\cdot N)},
\end{equation*}
where $C'$ is the first derivative of the cost function $C$, $\gamma$ is some positive constant, and $h(N), N\in\mathbb{N}$, is a  positive non-decreasing function in $N$.

By the convexity of the optimization problem \eqref{eq:CP_avg}, and the compactness of the feasible set, there exists an optimal solution $\vxs$ with $\xs_{n,t+1}(m)$ being the {\bf optimal} proactive download of data item $m$ made to user $n$ in slot $t$.
% \begin{multline*}
%\cQ_t(m):=\left\{n:\xs_{n,t}(m)>0\right\}, \\ 
%\quad t=0,\cdots, T-1, \quad m=1,\cdots, M, 
% \end{multline*}
% and
\begin{definition}[Active users]
\label{def:active}
For each data item $m$ and time slot $t$, we define a set $\cB_t(m)$ of active users as 
\begin{multline*}
\cB_t(m):=\left\{n:\E\left[\Int(m)C'(L_t)-C'(L_{t-1})\right]>0\right\}, \\
t=0,\cdots, T-1, \quad m=1,\cdots, M,
\end{multline*}
with $B_{t}(m):=|\cB_{t}(m)|$ is the cardinality of set $\cB_{t}(m)$.
\end{definition}
{The active users of slot $t$ w.r.t. item $m$ are those users who, in the expected sense, create {higher marginal cost} in slot $t$ by requesting item $m$ than the {marginal cost} of the previous slot. Thus, they have a high potential to improve the cost reduction through a proactive service of their demand.}  %In the definition above, the expectation
%\begin{equation*}
%\E\left[\Int(m)C'(L_t)-C'(L_{t-1})\right]
%\end{equation*}
% captures the \emph{marginal} contribution of user $n$ to the cost of time slot $t$, when requests item $m$, over the cost of the previous time slot $t-1$. {This, for instance, may attribute to the disparate activity levels of the users between times $t-1$ and $t$.}
  % The active users for any given slot have a high potential to improve the cost reduction through a proactive service of their demand.

 Despite the convexity of \eqref{eq:CP_avg}, a {\emph{closed form}} expression for the optimal value is not available for the general cost function defined above. As such, we study the asymptotic performance of $\Delta C(N)$ through upper and lower bounds that exhibit the same scaling order with $N$.

\subsection{Upper Bound}
We use the set of active users $\cB_t(m)$ to characterize an upper bound on $\Delta C(N)$ as follows.
\begin{lemma}%[Upper bound on cost reduction]
\label{lem:CR_UB}
Let $N\in\mathbb{N}$. For $\cB_{t}(m)$ defined above,
\begin{equation}
\label{eq:CR_upper}
\begin{split}
\Delta C(N)\leq &\\
 \frac{1}{T}\sum_{t=0}^{T-1}\sum_{m=1}^{M}S(m)\sum_{n\in\cB_{t}(m)}\E\left[\Int(m)C'(L_t)-C'(L_{t-1})\right].
\end{split}
\end{equation}
\end{lemma}
 \begin{proof}
Follows by the mean value theorem (MVT) for random variables. Full proof is provided in Appendix B of \cite{TR}.
\end{proof}

\subsection{Lower Bound}
In order to establish a lower bound on the cost reduction, we introduce the following preliminaries.
%\begin{definition}
For every time slot $t\in\{0, \cdots,T-1\}$, suppose that $\cB_{t}(m)$ is non-empty for some data item $m$. We define the quantity $\hat{x}_t$ as
\begin{multline}
\label{eq:x_hat}
%\begin{aligned}
\hat{x}_t:=\arg\min_{0\leq x_t\leq \Smin} \E\Biggl[C\Biggl(L_{t-1}+\sum_{m=1}^{M}\sum_{n\in\cB_{t}(m)}x_t\Biggr)\\
+C\Biggl(L_t-\sum_{m=1}^{M}\sum_{n\in\cB_{t}(m)}x_t\Int(m)\Biggr)\Biggr].
%\end{aligned}   
\end{multline}
%\end{definition}
{In other words, $\hat{x}_t$ is the optimal proactive download to all active users at time $t$, if proactive downloads are to be equally and exclusively assigned to the active users, while omitting the contribution of the proactive downloads from the previous slot.} %{The reason for omitting $x_{t-1}$, is to decouple the choice of $\hat{x}_t$ from $\hat{x}_{t-1}$, and hence facilitate the construction of lower and upper bounds.}

The following result about $\hat{x}_t$ holds.

\begin{lemma}
\label{lem:+x_hat}
If $\cB_t(m)$ is non-empty for some time slot $t$ and data item $m$, then $\hat{x}_t>0$.
\end{lemma}
%\begin{proof}
%Please refer to Appendix \ref{app:+x_hat}.
%\end{proof}

\begin{definition}[Policy A]
\label{def:Policy_A}
A proactive download allocation strategy named Policy A produces a proactive download vector $\tilde{\vx}$ satisfying the constraints of \eqref{eq:CP_avg} as follows.
For time slot $t=0,\cdots,T-1$, and any data item $m$,
\begin{equation}
\tilde{x}_{n,t}(m)=\begin{cases}
\tilde{x}_{t} & \text{ if } n\in\cB_{t}(m),\\
0           & \text{ if } n\notin\cB_{t}(m),
\end{cases}
\end{equation}
where
\begin{equation}
\label{eq:x_tilde}
\tilde{x}_t:=\hat{x}_t-r,   
\end{equation}
for some $r>0$ chosen such that\footnote{Such a positive $r$ exists since $T$ is finite, and by Lemma \ref{lem:+x_hat}, $\hat{x}_t>0$ for any non-empty $\cB_{t}(m)$.} $\tilde{x}_t>0$ for all $t\in\{0,\cdots,T-1\}$, and $\hat{x}_t$ is defined in \eqref{eq:x_hat}.
\end{definition}

Note that, Policy A, assigns equal proactive downloads to all \emph{active users} in a given slot $t$ {while the use of $\tilde{x}_t$ ensures the positivity of the lower bound established below}. We utilize this policy in establishing the following lower bound on the performance of $\Delta C(N)$.

\begin{lemma}[Lower bound on cost reduction]
\label{lem:CR_LB}
Let $N\in\mathbb{N}$. Under Policy A and for $\cB_t(m)$ defined above, 
\begin{multline}
\label{eq:CR_lower}
%\begin{split}
 \Delta C(N) \geq \frac{1}{T}\sum_{t=0}^{T-1}\sum_{m=1}^{M}\tilde{x}_t\sum_{n\in\cB_{t}(m)}\\ \E\Biggl[\Int(m)C'\Biggl(L_t-\sum_{j=1}^{M}\sum_{k\in\cB_t(j)}\tilde{x}_t\I_{k,t}(j)\Biggr)\\
  \quad \quad -\quad C'\Biggl(L_{t-1}+\sum_{j=1}^{M}\sum_{k\in\cB_{t}(j)}\tilde{x}_t\Biggr)\Biggr] > 0.
%\end{split} 
\end{multline} 
\end{lemma}
\begin{proof}
Proof is omitted for brevity. We refer interested readers to Appendix D in \cite{TR}. %Please refer to Appendix \ref{app:CR_LB}.
\end{proof}
 {Thus, Lemma \ref{lem:CR_LB} confirms that the suboptimal policy A, which assigns equal proactive download values to active users, yields a strictly positive cost reduction, that will be shown in Lemma \ref{lem:positive_liminf} to attain an optimal asymptotic scaling with the number of users.}
 Having established general upper and lower bounds on the potential cost reduction, we study its asymptotic performance with the number of users and present the result in the next subsection.

\subsection{Asymptotic Analysis}
\label{subsec:asymptotic}
When the number of users $N$ grows to infinity, the expected time average cost also grows to infinity when for some time slot $t$, all users request a data item with a positive probability. The cost reduction itself will also grow to infinity with a certain scaling order with $N$. Such a cost reduction depends mainly on the number of elements in $\cB_t(m)$, and how it scales with $N$. {Since such number is system specific, and out of our design, we will address different cases of scaling of the number of active users with the total number of users. In each case, we will characterize the optimal order of growth of the cost reduction.} Throughout this subsection, we assume that $\qnt<1-\epsilon$ for all $n,t$ and some $\epsilon>0$. That is, each user can request a data item at any time slot with a positive probability, and this probability will not vanish as the number of users grows to infinity.   

The following assumption considers the asymptotic behavior of $B_{t}(m)$ as $N\to\infty$.
\begin{assumption}[Scaling of the number of active users]
\label{as:lim_exists}
Assume that there exists some non-decreasing function $h:\mathbb{N}\to\mathbb{N}$ such that $h(N)\to\infty$ as $N\to\infty$, and for every time slot $t$ and data item $m$, the limit
\begin{equation}
\beta_t(m):=\lim_{N\to\infty}\frac{B_{t}(m)}{h(N)}
\end{equation}
exists, where {$0\leq \beta_t(m)<\infty$}, $\forall m=1,\cdots,M, t=0,\cdots,T-1$.
\end{assumption}
Thus, the function $h(N)$ captures the maximum possible scaling of the active users with the total number of users $N$.
%{Such a scaling can by illustrated by the dramatically increasing peak-hour demand with the number of users, while the off-peak-hour demand does not increase in the same order.} 
Note that, under Assumption \ref{as:lim_exists}, $\beta_t(m)>0$ implies the number of active users grows to infinity as the total number of users does. The following two lemmas are crucial to establish the main asymptotic scaling result.
\begin{lemma}
\label{lem:bounded_limsup}
Under Assumption \ref{as:lim_exists}, there exists a positive constant $\gamma_2$ such that
\begin{equation*}
\begin{aligned}
\rho_t(m):&=\limsup_{N\to\infty}\sum_{n\in\cB_{t}(m)}\frac{\E\left[\Int(m)C'(L_t)-C'(L_{t-1})\right]}{h(N)C'(\gamma_2\cdot N)}\\
& <\infty, \quad \forall  m,t.
\end{aligned}
\end{equation*}
Further, if $\beta_t(m)=0$, then $\rho_t(m)=0$.
\end{lemma}
\begin{proof}
Proof is omitted for brevity. We refer interested readers to Appendix E in \cite{TR}.%Please refer to Appendix \ref{app:bounded_limsup}.
\end{proof}

\begin{lemma}
\label{lem:positive_liminf}
Under Assumption \ref{as:lim_exists}, suppose that $\sum_{m=1}^{M}\beta_t(m)>0$ for some time slot $t$, and that
\begin{equation}
\label{eq:delta_condition}
\liminf_{N\to\infty} \frac{\sum_{m=1}^{M}\sum_{n\in\cB_t(m)}\E\left[C'(L_t)\right]}{\sum_{m=1}^{M}\beta_t(m)\E\left[C'(L_{t-1})\right]}>1, \text{ for the same slot } t.
\end{equation}
Then,
\begin{equation*}
\begin{aligned}
 \sigma_t &:=\liminf_{N\to\infty}\frac{1}{h(N)C'(\gamma_1\cdot N)} \sum_{m=1}^{M}\sum_{n\in\cB_{t}(m)} \\
 & \quad \E\left[\Int(m)C'\left(L_t-\sum_{j=1}^{M}\sum_{k\in\cB_{t}(j)}\tilde{x}_t\I_{k,t}(j)\right)\right]- \\  & \E\left[C'\left(L_{t-1}+\sum_{j=1}^{M}\sum_{k\in\cB_t(j)}\tilde{x}_t\right)\right]>0, \text{ for some } \gamma_1>0.
\end{aligned}
\end{equation*}
\normalsize
\end{lemma}
\begin{proof}
Please refer to Appendix \ref{app:positive_liminf}.
\end{proof}
Condition \eqref{eq:delta_condition} implies that the average marginal cost due to the active users at time $t$ grows with $N$ faster than $\E[C'(L_{t-1})]$. {Such condition can particularly be realized in huge events and rallies, e.g., football games, whereby fans turn to be active in the breaks and aggressively fetch game highlights. Hence, proactive downloads of these highlights during the game play can significantly reduce service costs. } {The intuition of Condition (14) can be illustrated as follows. The contribution of active users to the marginal cost of a time slot $t$ should not vanish relative to the marginal cost of the previous slot $t-1$ in the asymptotic scenario, thus proactive downloads can enable a positive gain that is proportional to $h(N)C'(N)$ as proven in Theorem \ref{th:m_result}. Yet, if such condition does not hold, there will still be a positive gain to reaped, but will be shown in Theorem \ref{th:finite_B} to scale as $C'(N)$.}% Hence, the contribution of the active users to the cost at time $t$ is at least one time larger than the marginal cost at time $t-1$. The time slot $t$ can be interpreted as a peak-hour scenario with essentially excessive demand as compared to the preceding off-peak hour $t-1$.

Lemmas \ref{lem:bounded_limsup}, and \ref{lem:positive_liminf} are instrumental to establish the asymptotic scaling result of the cost reduction, which is stated in the following theorem.
\begin{theorem}[Asymptotic cost reduction for infinite number of active users]
\label{th:m_result}
Under Assumption \ref{as:lim_exists}, suppose that the cost function $C$ satisfies Condition \eqref{eq:delta_condition}. Then, there exist finite positive constants $\gmax$, and $\gmin$ for which
\begin{equation}
\label{eq:lim_sup}
\limsup_{N\to\infty} \frac{\Delta C(N)}{h(N)\cdot C'(\gmax\cdot N)}\leq \frac{1}{T}\sum_{t=0}^{T-1} \sum_{m=1}^{M} S(m) \rho_t(m),
\end{equation}
and
\begin{equation}
\label{eq:lim_inf}
\liminf_{N\to\infty}\frac{\Delta C(N)}{h(N)\cdot C'(\gmin\cdot N)} \geq \frac{1}{T}\sum_{t=0}^{T-1} \sum_{m=1}^{M} \chi_t\sigma_t(m),
\end{equation}
where $\chi_t:=\liminf_{N\to\infty}\tilde{x}_t$, $t=0,\cdots,T-1$.
\end{theorem}
\begin{proof}
The proof is straightforward from Lemma \ref{lem:bounded_limsup}, and Lemma \ref{lem:positive_liminf}.
\end{proof}
%Theorem \ref{th:m_result}, characterizes asymptotic upper and lower bounds on the scaling of the cost reduction leveraged through proactive data download. Moreover,
We can also conclude the following.
\begin{corollary}
\label{cor:inf_B}
Under Assumption \ref{as:lim_exists}, and Condition \eqref{eq:delta_condition}, if for some time slot $t$ and data item $m$, $\beta_t(m)>0$, then there exists a finite positive constant $\gamma$ such that
\begin{equation}
\label{eq:Theta_1}
\Delta C(N)=\Theta\left(h(N)C'(\gamma N)\right).
\end{equation}
\end{corollary}
Thus, if the number of active users grows to infinity as $h(N)$, the leveraged cost reduction grows unboundedly to infinity as $h(N)\cdot C'(\gamma N)$. At this point, the following remarks can be made.

\begin{remark}
For a polynomial cost function with some degree $d>1$, the leveraged cost reduction scales with the number of users as $h(N) N^{d-1}$. Further, if $h(N)$ is linear in $N$, then $\Delta C(N)$ grows as $N^d$, i.e., $\Delta C(N)$ scales with $N$ as the cost itself does.
\end{remark} 

\begin{remark}
\label{rk:exp}
For an exponential cost function, the cost reduction grows exponentially with the number of users. However, the constant $\gamma$ affects the exponent of scaling.
\end{remark}

%\begin{remark}
%The realization of Condition \eqref{eq:delta_condition} requires a class of super-linear cost functions. In particular, for each cost function $C$, a positive constant $\delta>0$ has to exist and satisfy
%\begin{equation*}
%\lim_{L\to\infty}\frac{L^{\delta}}{C'(L)}=0.
%\end{equation*} 
%This is manifested by considering the following setup. Let $C(L)=L-\log(L+1)$, which is a monotonically increasing and strictly convex cost function, but does not obey Condition \eqref{eq:delta_condition}. Assuming $T=2$, $L_0=0$, almost surely, and $L_1=SN$ for a single-data-item system with size $S$. Given this setup, all users are active in slot $1$, however, the obtained optimal cost reduction does not scale as $\Theta(\frac{\gamma N^2}{1+\gamma N})$ for any $\gamma>0$.
%\end{remark}

%\begin{remark}
%The super-linearity of the class of cost functions satisfying discussed in the previous remark specifies a typical form of practical types of costs, such as those measuring delays and energy consumption in communication networks. 
%\end{remark}

Until this point, we have addressed the scenario where the number of active users for some data items and time slots grows to infinity {which has been formalized in Lemmas \ref{lem:bounded_limsup}, \ref{lem:positive_liminf}, and Corrollary \ref{cor:inf_B}.} Now, we investigate the potential of scaling cost reduction when the number of active users {does not necessarily satisfy the conditions specified in those Lemmas.} Such a case can essentially result from a high uncertainty about the exact future demand for each user, especially when the user preferences are equally distributed over the set of data items. {In particular, we will show that the leveraged gain does not depend on $h(N)$, yet it will only scale as the first derivative of the cost function.} The following assumption captures this scenario.
\begin{assumption}[Finite number of active users]
\label{as:finite_B}
Assume $\liminf_{N\to\infty}B_t(m)>0$ for all $m,t$, and that $B_t(m)\leq B$ for some $B\in [0,\infty)$, and for all $m=1,\cdots,M$, $t=0,\cdots,T-1$, and $N\in\mathbb{N}$.
\end{assumption}
Assumption \ref{as:finite_B} formalizes the case when there always exists a positive number of active users in the systems, but such a number is bounded while  $N$ grows to infinity. The scaling order of the cost reduction with the number of users, under that assumption, is characterized in the following theorem.

\begin{theorem}[Asymptotic cost reduction for bounded number of active users]
\label{th:finite_B}
Under Assumption \ref{as:finite_B}, let
\begin{align*}
\cBmax_t(m)&:=\limsup_{N\to\infty}\cB_t(m),\\
\cBmin_t(m)&:=\liminf_{N\to\infty}\cB_t(m), \quad \forall m,t.
\end{align*}
The expected time average cost $\Delta C(N)$ satisfies
\begin{equation}
\label{eq:upper_finite}
\begin{split}
\limsup_{N\to\infty}\frac{\Delta C(N)}{C'(\gmax\cdot N)}\leq \frac{1}{T}\sum_{t=0}^{T-1}\sum_{m=1}^{M}S(m)\sum_{n\in\cBmax_t(m)}\rho_{n,t}(m),
\end{split}
\end{equation}
\begin{equation}
\label{eq:lower_finite}
\liminf_{N\to\infty}\frac{\Delta C(N)}{C'(\gmin\cdot N)}\geq \frac{1}{T}\sum_{t=0}^{T-1}\sum_{m=1}^{M}\chi_t\sum_{n\in\cBmin_t(m)}\sigma_{n,t(m)},
\end{equation}
for some finite positive constants $\gmax, \gmin$,
where
\begin{equation*}
\rho_{n,t}(m):=\limsup_{N\to\infty}\frac{\E\left[\Int(m)C'(L_t)-C'(L_{t-1})\right]}{C'(\gmax\cdot N)},\; n\in\cBmax_t(m),
\end{equation*}
\begin{multline*}
%\begin{split}
\sigma_{n,t}(m):=\\
\liminf_{N\to\infty}\frac{\E\left[\Int(m)C'\left(L_t-\sum_{j=1}^{M}\sum_{k\in\cB_{t}(j)}\tilde{x}_t\I_{k,t}(j)\right)\right]}{C'(\Smax N)}\\ -\frac{\E\left[C'\left(L_{t-1}+\sum_{j=1}^{M}\sum_{k\in\cB_t(j)}\tilde{x}_t\right)\right]}{C'(\Smax N)},\quad  n\in\cBmin_t(m).
%\end{split}
\end{multline*}
\end{theorem}
\begin{proof}
The proof follows by taking $\limsup$ and $\liminf$ of \eqref{eq:CR_upper} and \eqref{eq:CR_lower}, respectively, as $N\to\infty$.
\end{proof}
{Theorem \eqref{th:finite_B} establishes the fact that the leveraged cost reduction even under a finite number of active users still grows as $C'(\gmin N)$ with the total number of users. The scenario of finitely many active users represents a worst case on the user behavior whereby the dominant population of users are indifferent to data items, thus can not be categorized as \emph{active users}. Yet, in Section \ref{sec:Demand_Shaping}, we argue that these inactive users are more flexible to change their profiles to more deterministic ones in response to some incentives offered by the service provider.}
Note that, in Theorem  \ref{th:finite_B} if $\rho_{n,t}(m)>0$ for some $n,t,m$, then $\chi_t>0$,  and $\sigma_{n,t}(m)>0$, $\forall m,n,t$. This holds since the terms $$\displaystyle{\sum_{j=1}^{M}\sum_{k\in\cB_{t}(j)}\tilde{x}_t\I_{k,t}(j)}, \text{ and } \displaystyle{\sum_{j=1}^{M}\sum_{k\in\cB_t(j)}\tilde{x}_t}$$ are bounded almost surely by Assumption \ref{as:finite_B}, while the terms $L_{t-1}$, and $L_t$ grow to infinity almost surely as $N\to\infty$, for all $m,n,t$.

\begin{corollary}
\label{cor:finite_B}
Under Assumption \ref{as:finite_B}, if the RHS of \eqref{eq:upper_finite} is positive, then there exists a positive constant $\gamma$ such that
\begin{equation}
\label{eq:s_result}
\Delta C(N)=\Theta(C'(\gamma N)).
\end{equation}
\end{corollary}

The significance of \eqref{eq:s_result} is manifested under the class of cost functions with $C'(N)\to\infty$ as $N\to\infty$, which yield unbounded cost reduction for a bounded number of active users. The following remark highlights this gain.

\begin{remark}
For a class of cost functions with $C'(N)\to\infty$ as $N\to\infty$, the leveraged cost reduction grows unboundedly to infinity as $C'(\gamma N)$ even if there is only \textbf{one} active user for only \textbf{one} data item at only \textbf{one} time slot. %The reason behind such unbounded gain attributes to the coupling of the cost incurred to serve the request of this active user and the costs of serving the rest of users, which grows unboundedly with $N$.
\end{remark}

Theorem \ref{th:finite_B} and Corollary \ref{cor:finite_B} have provided a worst-case scaling scenario for the cost reduction. In the asymptotic scenario of this case, the demand profiles of all but finitely many users are rather indeterministic and creating a substantial uncertainty about the expected user demand. Such a confusion may limit the proactive downloads for a wide set of users.  {In order to tackle such a problem, in \cite{ISIT13}, we have considered the scenario whereby service providers utilize caching centers to proactively fetch prospective content and supply it to interested users upon demand. However, the notion of demand shaping to be presented in the next section is projected to leverage even more cost reduction through}, \emph{joint} allocation of user demand profile and proactive data download, if users allow for modifying their demand profiles. 

\section{Demand Shaping}
\label{sec:Demand_Shaping}
Our notion of demand shaping is motivated by the observation that less {deterministic} users, whose demand profiles are divided almost equally over a subset of data items, can be expected to be more \emph{responsive} to recommendation differences offered by the service provider. For example, a user who is indifferent to watching a documentary on the American revolution or the civil war is likely to be more responsive to the recommendation disparity between the two. We highlight the following point. Our proposed framework uses recommendations to improve the {certainty} of future demands, hence the accuracy of the proactive downloads. Therefore, less {deterministic} users that are sensitive to data item valuation differences provide a high potential for proactive service gains.

The joint proactive download and demand shaping problem is given by $C^{\cJ}(\vp^*,\vxs):= $
\begin{equation}
\label{eq:Joint}
\begin{aligned}
 &  \underset{(\vp,\vx)}{\min}  \quad   \Co\left(L_t+\sum_{m=1}^{M}\sum_{n=1}^{N}x_{n,t+1}(m)-x_{n,t}(m)\Int(m)\right)\\ 
& \text{subject to}  \quad 0\leq x_{n,t}(m) \leq S(m), \quad \forall m,n,t,\\
  & \quad \quad \vp_{n,t}\in\cF_{n,\vtp_{n,t}}, \quad \forall n,t,\\
\end{aligned}  
\end{equation}
where $\vp=(\vp_n)_{n=1}^{N}$, and $\vp_{n}=(\vpnt)_{t=0}^{T-1}$.

We attack the joint design problem in two scenarios due to their structural differences. First, we consider a relaxed optimal solution, a scenario whereby no restrictions are imposed on the new users' profiles other than maintaining the same probability of being inactive, $\qnt$, $n=1,\cdots,N$, $t=0,\cdots,T-1$. Afterwards, we study a more practical scenario where users' satisfaction constraints are imposed.

\subsection{Jointly Optimal Solution for Fully Flexible Users}
\label{sec:w_constraint}
In order to gain insights on the structure of the best user profile leading to minimum expected cost with proactive downloads, we consider the relaxed version of the optimization problem \eqref{eq:P_f_horizon} satisfying:
\begin{assumption}
\label{as:relaxed}
Assume that user demands are fully flexible, i.e., 
\begin{multline*}
\cF_{n,\vtpnt}=\left\{\vpnt:\sum_{m=1}^M \pnt(m)=1-\qnt, \pnt(m)\geq 0 \right\}.
\end{multline*}
\end{assumption}
%Under Assumption \ref{as:relaxed}, the joint proactive download and demand profile allocation problem is formulated as  
 Note that, the new user profile is chosen such that the probability of a user $n$ remains silent in slot $t$ is $\qnt$ \textbf{unchanged}, hence users do not have to change their activity pattern. %Thus, the service provider may modify the preferences of the users through data item valuation techniques, but it can not change the activity of each user of whether to request a data item at all or remain silent.

% Note that in this section, we focus only on \emph{minimizing} the expected time average cost incurred by the service provider under proactive downloads instead of the expected cost reduction. The key idea is that, maximizing cost reduction through changing the user profiles affects the cost of the non-proactive network as well which renders the comparison unfair. In Section \ref{sec:pro_down}, however, proactive downloads only do not affect the cost of the non-proactive network, which validates the use of the cost reduction metric.

Under Assumption \ref{as:relaxed}, since the feasible set of \eqref{eq:Joint} is compact and the objective function is  convex, then there exists a globally optimal solution to the problem \cite{Boyd}. We denote such a solution by $(\vps,\vxs)$ and characterize it in the following theorem.

\begin{theorem}
\label{th:joint}
Under Assumption \ref{as:relaxed}, define $\cM^*:=\{m:S(m)=\Smin\}$ and pick some $m^*\in\cM^*$, then 
\begin{equation}
\label{eq:pjoint}
\pnt^*(m)=\begin{cases} 1-\qnt, & m=m^*,\\
0, & m\neq m^*,
\end{cases}
\end{equation}
% Moreover, let 
% \begin{equation*}
% \Int^*(m):=\begin{cases} 1,& \text{\upshape{ with probability }} \pnt^*(m),\\
% 0 & \text{\upshape{ with probability }} 1-\pnt^*(m),
% \end{cases}
% \end{equation*}
% then
%\small
\begin{equation}
\label{eq:xjoint}
\begin{aligned}
\vxs=&\arg \min_{\vx} \Co\left(L^*_t+\sum_{m=1}^{M}\sum_{n=1}^{N}x_{n,t+1}(m)-x_{n,t}(m)\Int^*(m)\right)\\  
&\text{{subject to}}
 \quad 0\leq x_{n,t}(m) \leq S(m), \quad \forall m,n,t,
\end{aligned}  
\end{equation}
\normalsize
where $\I^*_{n,t}(m)$ is the indicator function associate with $\pnt^*(m)$, and $L^*_t$ is the non-proactive network load under $\vp^*$.
Further, if $|\cM^*|=1$, then $(\vps,\vxs)$ is {\bf unique}.
\end{theorem}

\begin{proof}
Please refer to Appendix \ref{app:joint}.
\end{proof}
Theorem \ref{th:joint}, therefore, suggests a new user profile that ensures requesting the data item with the least size if the user is to request a data item at all. %Then proactive downloads under the new profile minimize the average incurred cost for service. 
In the next subsection, we model and study the problem under user \emph{satisfaction} constraints.

\subsection{Joint Design under User Satisfaction Constraints}
\label{sec:s_constraint}
We propose a model for capturing the economic responsiveness and service flexibilities of the user demands, which will facilitate the design. Suppose that $\vtpnt$ is the initial (given) profile of user $n$ for slot $t$, with $\qnt$ being the probability that user $n$ remains silent at slot $t$. 
%\begin{definition}
We utilize the notation
\begin{equation*}
\tpint(m):=\frac{\tpnt(m)}{1-\qnt}, \quad \forall m,n,t,
\end{equation*}
to denote the probability that user $n$ requests data item $m$ at time slot $t$, given that he decides to request a data item at all. We also use $\vtpint:=(\tpint(m))_{m=1}^{M}$.
%\end{definition}
%\subsubsection{Entropy Ball Constraint}

Next, we introduce a key measure and a related constraint that capture the economic responsiveness and flexibility of users in shifting their demand profile within acceptable service quality limits. 

\begin{definition}(Entropy Ball Constraint (EBC))
\label{def:EBS}
For user $n$ at time slot $t$ with initial profile $\vtpnt$ and probability of being silent $\qnt$, we say that a new profile $\vpnt$ satisfies the {\bf entropy ball constraint} if $\sum_m \pnt(m)=1-\qnt$ and
\begin{equation}
\label{eq:EBS}
\frac{\norm{\vtpnt-\vpnt}}{1-\qnt} \leq \an H(\vtpint), \quad \forall n,t,
\end{equation}
where $\an$ is a positive constant differentiating between user classes and normalizing the right hand side of \eqref{eq:EBS}, and $H(\vtpint)=-\sum_m \tpint(m)\log \tpint(m)$ is the entropy \cite{IT} of $\vtpint$.
\end{definition}  

The above metric utilizes the entropy of the choice of user $n$ under the initial profile $\vtpnt$ to capture the radius of an  $M-$dimensional ball centered at $\vtpint$. The reason behind the use of the entropy function is that the higher the uncertainty about the user demand, the higher the entropy is. Hence, the user has higher potential to modify the demand profile. This holds since users with indeterministic demand may not mind a specific data item, hence grant the  service provider more flexibility to suggest modified profiles. 

On the other hand, small entropy reflects a deterministic demand profile whereby the user is not flexible to change his profile. Yet, being deterministic already facilitate efficient proactive downloads.% Moreover, we use the parameter $\an$ to further control the radius of the entropy ball as user $n$ can belong to a higher class where $\an$ is small resulting in a tight satisfaction region and perhaps an unchanged profile. Fig. \ref{fig:Entropy} depicts an illustrating shape of the EBC for $M=3$.

%\begin{figure}[ht]
%	\centering
%		\includegraphics[width=0.35\textwidth]{Entropy.eps}
%	\caption{An example of $M=3$ data items that user $n$ can choose from in slot $t$. Entropy ball shrinks to a zero radius at the corner points, and attains a maximum radius at the center point.}
%	\label{fig:Entropy}
%\end{figure}

Using the above EBC, the constraint set $\cF_{n,\vtpnt}$ can be written as
\begin{assumption}
\label{as:EBC_cons}
\begin{equation*}
\cF_{n,\vtpnt}=\{\vpnt:\vpnt \text{ satisfies EBC},\pnt(m)\geq 0, \forall m\}.
\end{equation*}
\end{assumption}

Note that EBC maintains a constant probability of remaining silent, $\qnt$, for all users and time slots.  Statistically, the user does not have to change his access rate.

%Under Assumption \ref{as:EBC_cons}, the \emph{joint proactive data download and demand profile allocation problem} becomes:
%\begin{equation}
%\label{eq:Joint}
%\begin{aligned}
% & \underset{(\vp,\vx)}{\min}  
% \quad   \Co\left(L_t+\sum_{m=1}^{M}\sum_{n=1}^{N}x_{n,t+1}(m)-x_{n,t}(m)\Int(m)\right)\\  
% & \text{subject to} \\
%& \frac{\norm{\vtpnt-\vpnt}}{1-\qnt} \leq \an H(\vtpint), \quad \forall n,t,\\
%& \text{constraints of }\eqref{eq:Cjoint}.
%\end{aligned}  
%\end{equation}

%\subsubsection{Enhanced Proactive Downloads and Demand Shaping}
Under Assumption \ref{as:EBC_cons}, while the constraints of \eqref{eq:Joint} are all convex in $(\vp,\vx)$, the objective function, denoted by $f_0(\vp,\vx),$ is non-convex. Thus, in contrast to the tractable structure under fully flexible demands of previous section, the characterization of a global optimal solution of \eqref{eq:Joint} is computationally intractable. Nevertheless, we next show that strict performance improvement over proactive downloads can still be guaranteed.

To see this, suppose that $\vtx$ is the optimal proactive download allocation obtained under the initial user profile $\vtp$, where $(\vtp,\vtx)$ does not satisfy the KKT conditions \cite{Boyd} of \eqref{eq:Joint}, then a point $(\vhp,\vhx)$ which satisfies $f_0(\vhp,\vhx)<f_0(\vtp,\vtx)$, as well as the KKT conditions of \eqref{eq:Joint}, can be obtained through iterative solutions to approximate convex problems.
\begin{lemma}
\label{lem:approx}
Let $\fh^{k}$ be a convex function in $(\vp,\vx)$ that replaces $f_0$ of \eqref{eq:Joint} at iteration $k$. Denote by $(\vp^{k-1},\vx^{k-1})$ the optimal solution to the resulting convex optimization problem at the $k-1^{st}$ iteration, $k=1,2,\cdots$. If
\begin{enumerate}
\item $\fh^{k}(\vp,\vx)\geq f_0(\vp,\vx)$ for all feasible $(\vp,\vx)$,
\item $\nabla \fh^{k}(\vp^{k-1},\vx^{k-1})=\nabla f_0(\vp^{k-1},\vx^{k-1})$,
\item $\fh^{k}(\vp^{k-1},\vx^{k-1})=f_0(\vp^{k-1},\vx^{k-1})$,
\end{enumerate}
$\forall k=1,\cdots$, then $f_0(\vp^{k-1},\vx^{k-1})>f_0(\vp^{k},\vx^{k}), \forall k$, and the sequence $\{(\vp^{k},\vx^{k})\}_k$ converges to a point $(\vhp,\vhx)$ which is a locally optimal solution to \eqref{eq:Joint}. 
\end{lemma}

The above lemma is a special case of Theorem 1 in \cite{Old} which aims at providing local optimal solutions to non-convex optimization problems.
\begin{corollary}
Starting from an initial condition $(\vp^0,\vx^0)=(\vtp,\vtx)$, a sequence of approximate functions $\{\fh^{k}\}$ generated as in Lemma \ref{lem:approx} and resulting in a KKT-satisfying point $(\vhp,\vhx)$ leads to $f_0(\vtp,\vtx)>f_0(\vhp,\vhx)$.
\end{corollary} 
In the following theorem, we suggest a general approximation to $f_0$ of \eqref{eq:Joint} at each new iteration $k$ that converges to a locally optimal solution.

\begin{theorem}
\label{th:gen_approx}
For $f_0$ being the objective function of \eqref{eq:Joint}, the approximate function  
\begin{equation}
\label{eq:approx_f}
\fh^{k}(\vp,\vx)=f_0(\vp^{k-1},\vx)+\sum_{m,n,t}\frac{\partial f_0(\vp,\vx^{k-1})}{\partial \pnt(m)}\biggr |_{\vp=\vp^{k-1}}\cdot\pnt(m)
\end{equation}
at iteration $k\geq 1$ is convex in $(\vp,\vx)$, further, the sequence of solutions to the problem resulting from replacing $f_0$ with $\{\fh^{k}\}_k$ converges to a locally optimal solution of \eqref{eq:Joint}. 
\end{theorem}
\begin{proof}
Please refer to Appendix \ref{app:gen_approx}.
\end{proof}

The approximate function \eqref{eq:approx_f} can now be used to replace $f_0$ of \eqref{eq:Joint} at iteration $k\geq 1$. Starting with $(\vp^0,\vx^0)=(\vtp,\vtx)$, the successive solutions to approximate optimization problems with $f_0$ being replaced by $\fh^{k}$ of \eqref{eq:approx_f} converges to a point $(\vhp,\vhx)$ with $f_0(\vhp,\vhx)<f_0(\vtp,\vtx)$. Further, the following conclusion can be made about $(\vhp,\vhx)$.

\begin{theorem}
\label{th:boundary}
Suppose that $\{\vp_{n,t}:\norm{\vtpnt-\vhpnt}<\an {1-\qnt} H(\vtpint)\}$ is a proper subset of $\{\vp_{n,t}:\vp_{n,t}\geq 0, \sum_{m=1}^{M}\pnt(m)=1-\qnt\}$, then the locally optimal solution $(\vhp,\vhx)$ of \eqref{eq:Joint} satisfies
\begin{equation}
\label{eq:boundary}
\frac{\norm{\vtpnt-\vhpnt}}{1-\qnt} = \an H(\vtpint), \quad \forall n,t.
\end{equation}
\normalsize
In other words, the locally optimal profile $\vhpnt$ lies on the boundary of the EBC region.
\end{theorem}

\begin{proof}
Please refer to Appendix I in \cite{TR}. %\ref{app:boundary}.
\end{proof}

Theorem \ref{th:boundary} motivates the design of efficient low complexity schemes for joint proactive download and demand shaping in large scale systems.

In \cite{TR}, the proof of Theorem \ref{th:boundary} does not depend on the structure of the constraint region $\cF_{n,\vtpnt}$ for any $n,t$. Consequently, any optimal solution to the formulation \eqref{eq:P_f_horizon} always yields modified demand profiles that lie on the \emph{boundary} of the constraint region $\cF_{n,\vtpnt}$. {As the service provider aims to maximally utilize the user flexibility by pushing it to the boundary, we note that each user has the ability to control his own flexibility region (e.g., through the parameter $\alpha_n$ in the EBC case) and the service provider should offer sufficient incentives to trade for flexibilities. The main purpose of presenting the demand shaping aspect of the problem is to reveal the potential for additional cost minimization to be leveraged by changing the users' preferences over data item rather than changing their activities over time.}

\subsection{Impact on Cost Reduction}
With the additional cost minimization capabilities offered by demand shaping, it is worthwhile to explore the effect of demand shaping on the cost reduction metric introduced in Section \ref{sec:pro_down}. In particular, we use the notation
\begin{equation*}
\Delta C^*(N)=C^{\cN}(N,\vtp)-C^{\cP}(N)
\end{equation*} 
to denote the \emph{maximum} possible cost reduction obtained through proactive content download and demand shaping. The terms $C^{\cN}(N,\vtp)$, $C^{\cP}(N)$ are defined in \eqref{eq:N_f_horizon} (assuming initial demand profile $\vtp$), \eqref{eq:P_f_horizon}, respectively.

By adding and subtracting  $C^{\cP}(N,\vtp)$ (defined in \eqref{eq:CP_avg}), we can write $\Delta C^*(N)=\Delta C(N)+\Delta C^{\cP}(N)$, where
\begin{equation}
\Delta C^{\cP}(N)=C^{\cP}(N, \vtp)-C^{\cP}(N)
\end{equation}
is the expected cost reduction gain reaped through demand shaping. 

In this subsection, we present an additional scaling result of the cost reduction  $\Delta C^{\cP}(N)$ with $N$.
\begin{theorem}[Demand shaping gain]
\label{th:DS_gain}
Suppose that there exists a time slot $\tz$ and a user $\nz$ for which: $\sup_n \{\tilde{q}_{n,\tz}\}<1-\epsilon$ for some $\epsilon>0$, and that $\cF_{\nz,\vtpntz}$ contains a profile $\vbpntz$ that satisfies
\begin{equation}
\label{eq:load_condition}
\sum_{m=1}^{M}(S_m-\tilde{x}_{\nz,\tz})(\tilde{P}_{\nz,\tz}(m)-\bar{P}_{\nz,\tz}(m))>0.
\end{equation}
Then, there exists a constant $\gamma>0$ such that
\begin{equation}
\liminf_{N\to\infty}\frac{\Delta C^{\cP}(N)}{C'(\gamma N)}>0.
\end{equation}
\end{theorem}
The quantity $\tilde{\bf x}_{n,t}=(\tilde{x}_{n,t}(m))_m$ is the optimal proactive download under the original demand profile $\vtpnt$.

\begin{proof}
Please refer to Appendix \ref{app:DS_gain}.
\end{proof}

It can be seen from Condition \eqref{eq:load_condition} that the expected load for user $\nz$ under $\vbpntz$ is strictly smaller than that under $\vtpntz$. {Thus, even if user $\nz$ has been found inactive in the proactive download step, and consequently has not received any proactive downloads, the demand shaping step can potentially increase the certainty about his future demand, render him active, and assign him positive proactive download. Yet, as long as the user can slightly shift his preferences to decrease his expected load in slot $\tz$, e.g., from YouTube with the majority of video content to Facebook with expected less load due to the text- and picture-based content, the resulting} gain in this scales with the total number of users as $C'(\gamma N)$, at least. Thus, if the number of users that satisfy \eqref{eq:load_condition} also grows with $N$ according to some function $u(N)$, then the leveraged gain scales as $u(n)C'(\gamma N)$.

Not only does the result established in Theorem \ref{th:DS_gain} add to the cost reduction gains beyond proactive downloads, it also provides a simple design metric, Condition \eqref{eq:load_condition}, that can achieve significant gains with expanding customer base.

The final step in the proposed framework concerns the modification of the users' profiles to the new ones obtained through constrained cost minimization. Towards this end, we consider a new recommendation system that sets slightly different ratings to data item from those given by the users' themselves.

\section{Data Item Recommendation Scheme}
\label{sec:Recom}
Upon the calculation of the locally optimal demand profile $\vhp$, the service provider has to assign new valuations, $\vv=(\vv_{n,t})_{n,t}$, so that the users {react to them by following the modified profiles}. However, following a user-satisfaction based techniques that rely only on the user preferences and valuations (cf. \cite{CF1}-\cite{CF6}), the service provider has originally assigned a rating of $\rnm$ for data item $m$ to user $n$ (which is independent of the dynamics of the varying network load), but the requirements of the new demand profile necessitate a valuation of $\vnt(m)$ to be assigned to data item $m$ and user $n$ and time slot $t$. We presume that the new rating vector $\vv_{n,t}$ must be as close as possible, in the Euclidean distance sense, to the original rating vector $\vr_n=(\rnm)_{m=1}^{M}$, while achieving the new demand profile. Fig. \ref{fig:Demand_shaping} depicts a block diagram of the proposed data item valuation scheme.
\begin{figure}[ht]
	\centering
		\includegraphics[width=0.45\textwidth]{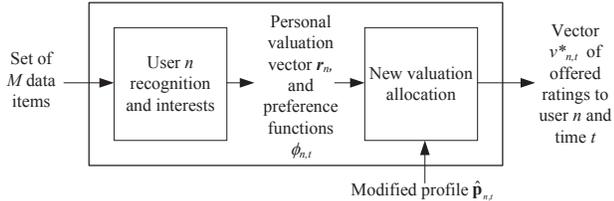}
	\caption{Block diagram of data item valuation system for a user $n$ at time slot $t$. The determination of the user interests applies through learning algorithms.}
	\label{fig:Demand_shaping}
\end{figure}

As has been hypothesized in Section \ref{sec:sys_mod}, there exists a function $\phi_{m,t}:[0,1]^M\to [0,1]$ that determines the user profile based on the offered ratings to the data items. Now, assuming that the initial profile for user $n$ at slot $t$ is $\vtpnt=(\tpnt(m))_{m=1}^{M}$, then according to $\phi_{m,t}$, we have $\tpnt(m)=\phi_{m,t}(\vr_n)$, where $\phi_{m,t}$ is supposed to be monotonically increasing in $\rnm$ and to measure the relative quality of data item $m$ to the rest of data items. After all, the rating allocation is a user-based problem where the service provider can compute new ratings for each user separately from the others. Denoting the new profile for user $n$ at time slot $t$ by $\vhpnt=(\hpnt(m))_{m=1}^{M}$, the valuation assignment problem for user $n$ at slot $t$ is formulated as
\begin{equation}
\label{eq:rating_f1}
\begin{aligned}
\underset{{\vv_{n,t}}}{\min} & \quad \norm{\vv_{n,t}-\vr_n} \\
\text{subject to}  &\quad \hpnt(m)=\phi_{m,t}(\vv_{n,t}),\quad \forall m,\\
  %                 & y_{m,n}\leq y_n^{max}, m=1,\cdots,N,\quad n=1,\cdots,N,\\
                   &\quad  \vv_{n,t}\in [0,1]^M.
\end{aligned}  
\end{equation}
\begin{remark}
Problem \eqref{eq:rating_f1} is convex if and only if $\phi_{m,t}$ is a linear fractional mapping in $\vv_{n,t}$ \cite{Boyd}.
\end{remark}

An example case on a linear fractional mapping is
\begin{equation}
\label{eq:phi}
\phi_{m,t}(\vv_n)=(1-\qnt)\frac{\vnm}{\sum_{j=1}^{M}v_{n}(j)}, \quad \forall m,n,t.
\end{equation}
In this example, the mapping function $\phi_{m,t}$ captures the relative preference of user $n$ to choose data item $m$ amongst all data items. While problem \eqref{eq:rating_f1} is convex, a globally optimal solution can be obtained efficiently through a gradient descent algorithm \cite{Boyd}. %A closed form expression for the optimal rating vector $\vv_n^*$ is not available though.

In the case where $\phi_{m,t}$ is \emph{not} a linear fractional mapping, the problem turns out to be non-convex, calling for approximate solutions. One possible way of handling such a difficulty is to replace the non-convex $\phi_{m,t}$ with an approximate linear fractional form, and iteratively solving for approximate solutions till convergence to a locally optimal rating vector, a similar approach to that used in Section \ref{sec:Demand_Shaping}.

\section{Numerical Simulations}
\label{sec:results} 
{We begin our simulations with a simple two-user example to illustrate the process of the proposed system, then show the gains for large number of users next.}
\subsection{Illustrating Example}
  {We consider a service provider} that serves  $M=3$ data items with fixed sizes ${\bf s}=(3,2,4)$. There are $N=2$ users that may request services on a daily basis, whereby the day is divided into $T=2$ time slots with demand profiles following a cyclostationary distribution with period $T$. One time slot is supposed to represent an off-peak hour demand with $\tilde{q}_{n,0}=1-\po$, $n=1,2$ with $\po=0.1$. The other time slot represents a peak-hour demand which has $\tilde{q}_{n,1}=1-\pp$, $n=1,2$. The profiles of both users during the off-peak hour are $\vp_{1,0}=\po\cdot(0.8,0.1,0.1)$, ${\vp}_{2,0}=\po\cdot(0.3,0.1,0.6)$. Likewise, during the peak hour ${\vp}_{1,1}=\pp\cdot(0.8,0.1,0.1)$, ${\vp}_{2,1}=\pp\cdot(0.3,0.1,0.6)$. The parameters $\po$ and $\pp$ represent the user activity during the off-peak and peak hours. We consider two main cost functions: (1) a quadratic cost function $C(L)=L^2, L\in\mathbb{R}_+$ and (2) an outage-constrained cost function $C(L)=\frac{L}{\mu-L}$, where $\mu$ is the maximum load that the service provider can afford at a given time slot. We use $\mu=9.8$ and plot the obtained results versus $\pp$ in Fig. \ref{fig:P_downloads}.  %In Figs. \ref{fig:E_C_Quad}, \ref{fig:E_C_Cap}, the expected total load is plotted versus $\pp$ under the two different cost functions. Both figures show an expected cost reduction leveraged from the proactive resource allocation, especially when the customer activity $\pp$ increases. 
{We can see from the figures that proactive downloads yield a significantly reduced cost that grows almost linearly with the peak hour activity $\pp$. On the other hand, the no-proactive-downloads scenario results in excessive costs that grow essentially superlinearly with $\pp$.}

 \begin{figure}[htp]
 \centering
  \subfloat[Expected cost for quadratic function vs. $\pp$.]{\label{fig:E_C_Quad}\includegraphics[width=0.25\textwidth]{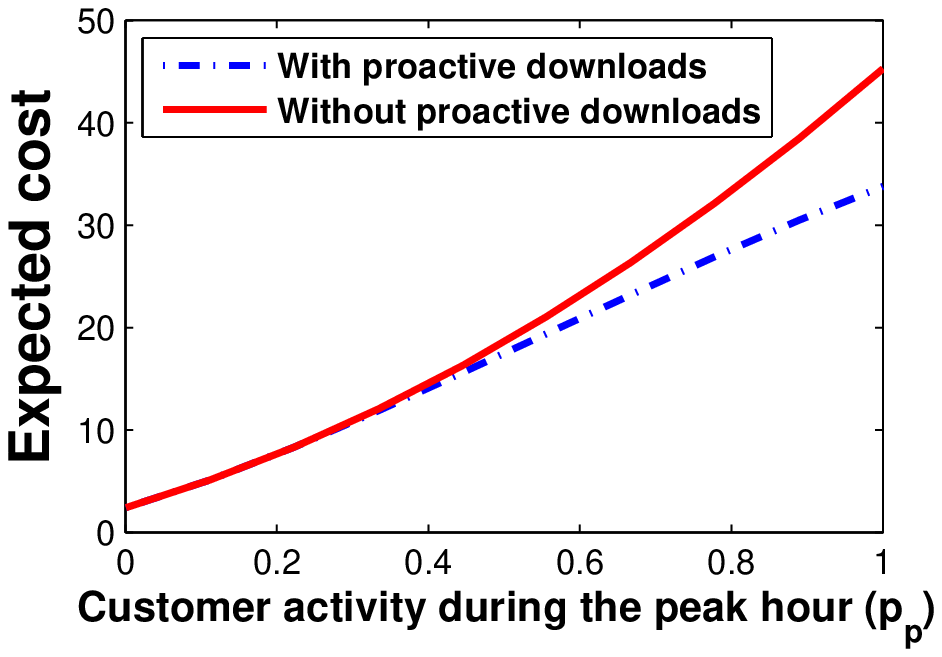}}
  \subfloat[Expected cost for capacity constrained function vs. $\pp$.] {\label{fig:E_C_Cap}\includegraphics[width=0.25\textwidth]{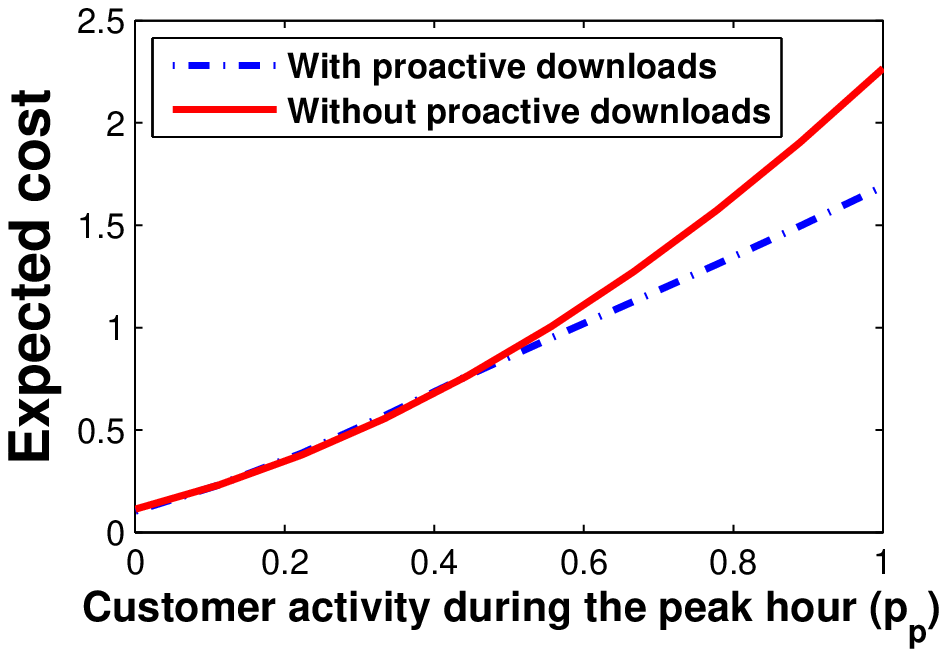}}\\
%  \subfloat[Expected load for quadratic function vs. $\pp$.] {\label{fig:E_L_Quad}\includegraphics[width=0.25\textwidth]{W_expected_load_quad.eps}}
%  \subfloat[Expected load for capacity constrained function vs. $\pp$.] {\label{fig:E_L_Cap}\includegraphics[width=0.25\textwidth]{W_expected_load_cap.eps}}
  \caption{Reduced cost under proactive downloads.}
  \label{fig:P_downloads}
\end{figure}
%In Figs. \ref{fig:E_L_Quad}, \ref{fig:E_L_Cap}, we show the total load during the off-peak and peak hours with and without proactive downloads. As can be observed, proactive downloads tend to pull the peak-hour load towards the low traffic period so as to minimize the difference between the two loads. However, because of the uncertainty about the customer choice, the allocation procedure does not exactly divide the load equally over the two traffic hours which may lead to unnecessary waste of resources and more cost. The upper histograms of both figures show the significant difference between the off-peak and peak hour loads, whereas the lower histograms reveal a considerably reduced gap between such loads, particularly in the quadratic cost function case. In the outage-constrained cost function case, the gap is not significantly reduced because of the system is fairly far from the instability point at which the load approaches the capacity. 
 
%\subsection{Optimal Profiles and Recommendations}

For the proposed joint user profile and proactive downloads scheme, we conduct a simulation with $\pp=0.9$ fixed. For both types of cost functions, the proposed iterative algorithm is {applied, taking} into account the approximate convex objective function developed in \eqref{eq:approx_f}. The convergence of results is plotted in Fig. \ref{fig:Joint_convergence} where it is clear that the proposed algorithms start from the initial proactive downloads and profiles $(\vtp,\vtx)$ and proceed with the iterative solutions until convergence to a local optimal solution to the original problem $(\vhp,\vhx)$. The resulting sequence of the objective functions (the cost functions) is strictly decreasing. 

 \begin{figure}[htp]
 \centering
  \subfloat[Quadratic function]{\label{fig:Conv_Quad}\includegraphics[width=0.25\textwidth]{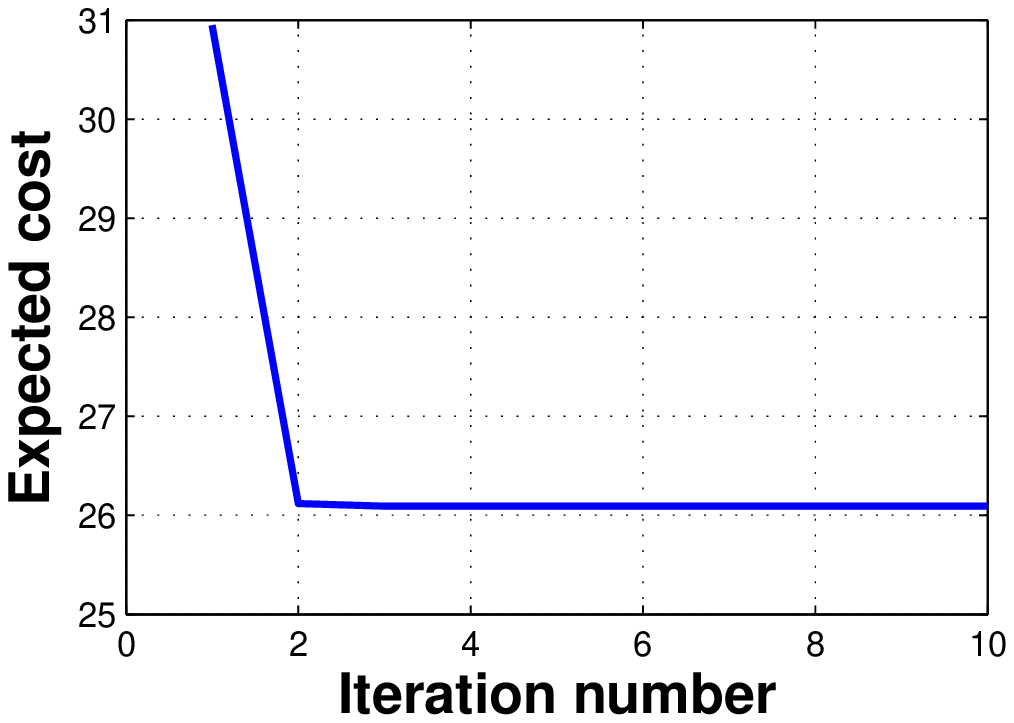}}  
  \subfloat[Outage-constrained function] {\label{fig:Conv_Cap}\includegraphics[width=0.25\textwidth]{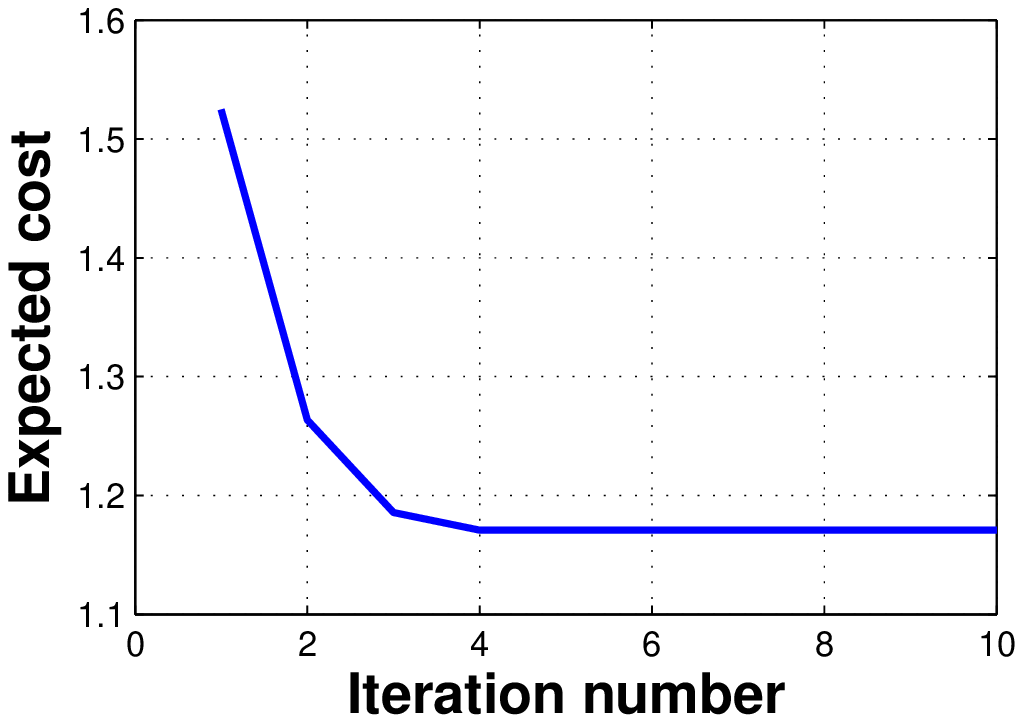}}
  \caption{Convergence of the joint user profile and proactive downloads allocation.}
  \label{fig:Joint_convergence}
\end{figure}

Upon obtaining the new peak-hour profiles $\vhp_{1,1}$, and $\vhp_{2,2}$, the recommendation process proposed in \eqref{eq:rating_f1} is invoked to assign new ratings for the available items, where we assume $\phi_{m,t}$ of \eqref{eq:phi}. We take the original valuation vectors as ${\bf r}_1=(0.8,0.1,0.1)$ and ${\bf r}_2=(0.3,0.1,0.6)$. The results of the simulations are summarized in Table \ref{tab:ResultsOfQuad}, \ref{tab:ResultsOfCapCons}, for the quadratic and outage-constrained cost functions respectively. It can be noted that, the modified profiles lie on the boundary of the entropy ball as the service provider is interested in pushing the profiles in the direction of the most deterministic behavior.% that leads to requesting the smallest expected cost. Also, the new ratings are reasonably close to the original.

\begin{table}[ht]
	\centering
		\begin{tabular}{| c| c| c| c| c|}

		\hline
	Item        & New profile  & New profile &  Rating    & Rating \\
		  index        & $\vhp_{1,1}/\pp$     & $\vhp_{2,1}/\pp$    &  ${\bf v}^*_{1,1}$ & ${\bf v}^*_{2,1}$  \\ \hline
		   1           &  0.8772           & 0.3111           & 0.7985         & 0.3381\\
		   2           &  0.1222           & 0.2211           & 0.1112         & 0.2403\\ 
		   3           &  0.0006           & 0.4678           & 0.0005         & 0.5084\\  \hline
		\end{tabular}
	\caption{\small Modified demand profiles and recommendations for quadratic cost function.}
	\label{tab:ResultsOfQuad}
\end{table}
\begin{table}[ht]
	\centering
		\begin{tabular}{| c| c| c| c| c| c| c| c|}

		\hline
	Item  & New profile  & New profile   &  Rating    & Rating \\
 index  			& $\vhp_{1,1}/\pp$     & $\vhp_{2,1}/\pp$    &  ${\bf v}^*_{1,1}$ & ${\bf v}^*_{2,1}$  \\ \hline
		   1      &  0.8298           & 0.3546             & 0.8005         & 0.2594 \\
		   2      &  0.1222           & 0.4507             & 0.1179         & 0.3297\\ 
		   3      &  0.0480           & 0.1947             & 0.0463         & 0.1424\\  \hline
		\end{tabular}
	\caption{\small Modified demand profiles and recommendations for a outage-constrained cost function.}
	\label{tab:ResultsOfCapCons}
\end{table}

\subsection{Increasing Number of Users}
To validate the scaling laws of the cost reduction, {we consider a simulation scenario of $M=50$ data items, with their size generated randomly from a uniform distribution in $[10:30]$. The period of cyclostationary demand profiles is $T=8$ slots. The demand profile for user $n$ and time $t$ follows a Zipf distribution with power $4$, that is, $\tpnt(m)=\frac{G_{n,t}}{m^4}$, where $G_{n,t}$ is a normalizing constant. We take $(\tilde{q}_{n,t})_n=\tilde{q}_t$, $\forall n$, with $(\tilde{q}_t)_t=(0.9,0.4,0.01,0.8,0.2,0.7,0.05,0.1)$. While considering a quadratic cost function $C(L)=L^2$, we compute the optimal cost reduction resulting from proactive downloads {\em only}. We plot the results in Figure \ref{fig:Performance_vs_N}. vs. the number of users $N$.} %we consider the following scenario {of $M=4$ data items, with sizes of $(S(m))_m=(10,8,14,12)$ units. The period of the cyclostationary demand profiles is $T=3$.  We take $(\tilde{q}_{n,t})_n=\tilde{q}_t$, $\forall n$, with $(\tilde{q}_t)_t=(0.2,0.9,0.01)$. The initial demand profiles for the users are generated as follows, $\tpnt(m)=\frac{(1-\tilde{q}_t)R(m)}{G_{n,t}}$, where $R(m)$ is an exponentially distributed random variable with mean $m$, and $G_{n,t}$ is a normalizing factor. While considering a quadratic cost function $C(L)=L^2$, and $\an=0.2, \forall n$, we run the proactive download and demand shaping algorithm for a varying number of users, average over 50 simulation runs, and plot  the results in Fig \ref{fig:Performance_vs_N}.}
%\begin{figure}[ht]
%	\centering
%		\includegraphics[width=0.35\textwidth]{Journal_asymptotic_CR.eps}
%	\caption{Asymptotic cost reduction $\Delta C(N)$ scales as $\Theta (N^2)$ under a quadratic cost function.}
%	\label{fig:Journal_asymptotic_CR}
%\end{figure}

%\begin{figure}[ht]
%	\centering
%		\includegraphics[width=0.35\textwidth]{Asymptotic_CR.eps}
%	\caption{\small Performance of cost reduction vs.  $N$.}
%	\label{fig:Asymptotic_CR}
%\end{figure}
%\begin{figure}[ht]
%	\centering
%		\includegraphics[width=0.35\textwidth]{Asymptotic_PL.eps}
%	\caption{\small Maximum average load vs.  $N$.}
%	\label{fig:Asymptotic_PL}
%\end{figure}

%\begin{figure}[htp]
% \centering
%  \subfloat[\small Cost reduction.]{\label{fig:Asymptotic_CR}\includegraphics[width=0.25\textwidth]{Asymptotic_CR.eps}}
%  \subfloat[\small Maximum average load.] {\label{fig:Asymptotic_PL}\includegraphics[width=0.25\textwidth]{Asymptotic_PL.eps}}\\
%  \subfloat[Expected load for quadratic function vs. $\pp$.] {\label{fig:E_L_Quad}\includegraphics[width=0.25\textwidth]{W_expected_load_quad.eps}}
%  \subfloat[Expected load for capacity constrained function vs. $\pp$.] {\label{fig:E_L_Cap}\includegraphics[width=0.25\textwidth]{W_expected_load_cap.eps}}
%  \caption{System performance vs. $N$.}
%  \label{fig:Performance_vs_N}
%\end{figure}

\begin{figure}[htp]
 \centering
  \subfloat[\small Cost reduction vs.  $N$.]{\label{fig:TON3_CR_asymptotic}\includegraphics[width=0.25\textwidth]{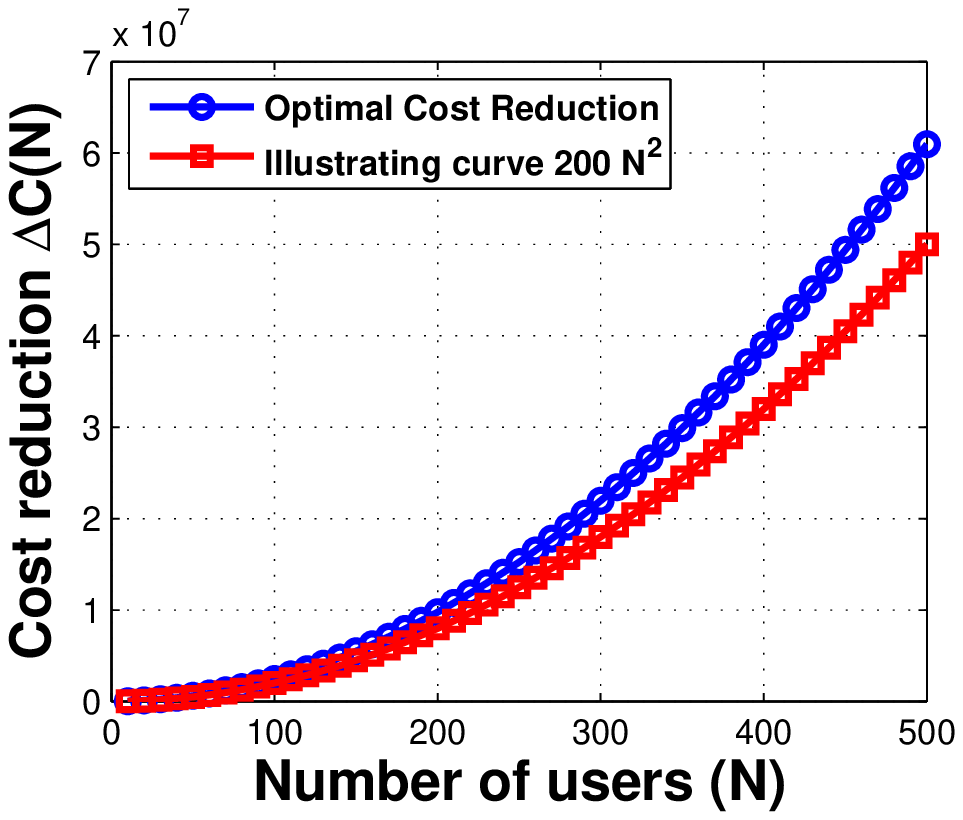}}
  \subfloat[\small Cost reduced by 17.7\%.] {\label{fig:TON3_CR_ratio}\includegraphics[width=0.25\textwidth]{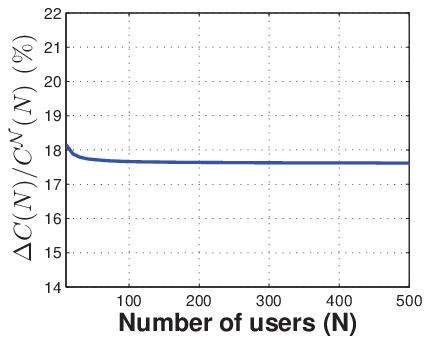}}
   \caption{Performance of cost reduction with number of users.}
    \label{fig:Performance_vs_N}
\end{figure}

%\begin{figure}[h]
%	\centering
%		\includegraphics[width=0.40\textwidth]{TON_asymptotic_LB_CR.eps}
%	\caption{Asymptotic scaling of cost reduction vs. number of users. The curve $N^2$ shown in red illustrates that $\Delta C=\Theta (C(N))$.}
%	\label{fig:TON_asymptotic_LB_CR}
%\end{figure}
In particular, Figure \ref{fig:TON3_CR_asymptotic} shows the leveraged cost reduction versus $N$, which manifests  the quadratically increasing gain with the number of users, where in the adopted simulation setup  $h(N)\geq 0.9 N$, $\forall N$, thus validates the theoretical result of Corollary \ref{cor:inf_B}. Further, Figure \ref{fig:TON3_CR_ratio} compares the leveraged cost reduction with that of the non-proactive network through the ratio $\frac{\Delta C(N)}{C^{\cN}(N)}$. The percentage gain is 17.7\%.
%{In Fig. \ref{fig:Asymptotic_CR}, cost reduction is depicted versus the number of users for proactive downloads only, and proactive downloads with demand shaping. As we computed the ratio $B_2(4)/N$ to be approximately $0.85$, we expect that the leveraged cost reduction grows with $N$ as $N^2$. This is illustrated by comparing the leveraged cost reductions with the curve $20N^2$ represented by red dashes on the figure. We can also note that demand shaping has improved the cost reduction performance by a factor of \%12.}

%{The average peak load of the system is computed and plotted versus $N$ in Fig. \ref{fig:Asymptotic_PL}. We can notice that the peak load of the unoptimized non-proactive system is about \%23.7 larger than that of proactive downloads, which reveals the potential for a smoothed-out traffic under proactive service. Moreover, with demand shaping, the peak expected load is further reduced by a factor of $\%8$.}

\section{Conclusion}
\label{sec:conc}
In this work, we have proposed and studied a proactive resource allocation and demand shaping framework for data networks. We aim to harness the predictability of future demand in creating more opportunities for a balanced load over time, hence a considerable resource utilization. We have considered the problem from the perspective of a service provider (SP) which incurs excessive costs during the peak-hour demand in order to sustain the service of the users' requests. We have proposed the notion of users'  \emph{demand profiles} to capture the statistical information about the future demand for each user. Such profiles are utilized in  \emph{proactive content downloads} where a portion of highly likely future demand is downloaded to the respective users during the off-peak hour, so as to minimize the \emph{time average expected cost}.

We have analyzed the asymptotic scaling laws of the cost reduction leveraged through such proactive downloads with the total number of users. We have proved the \emph{cost reduction} to scale similar to the expected cost of  the\emph{non-proactive} networks does. Even in the worst case scenario, where the network users are \emph{indeterministic}, the cost reduction scales with the first derivative of the cost function. In order to improve the certainty about the users' demand, we have proposed and studied the notion of \emph{demand shaping}, which is proved to strictly reduce the expected cost under user satisfaction constraints. We have developed a data-item \emph{recommendation} scheme that enhances the certainty about the demand of each user, hence the quality of proactive downloads. %We have validated the theoretical results with numerical examples to demonstrate the potential gains of the proposed framework.

\appendices

\section{Proof of Lemma \ref{lem:CR_LB}}
\label{app:CR_LB}

The cost reduction satisfies
\begin{align*}
& T \Delta C(N) \overset{(a)}{\geq} \sum_{t=0}^{T-1}\E\Biggl[C(L_t)-C\Biggl(L_t-\sum_{m,n\in\cB_{t}(m)}\tilde{x}_t\Int(m) +  \\
\end{align*}
\begin{align*}
&\sum_{m,n\in\cB_{t+1}(m)}\tilde{x}_{t+1}\Biggr)\Biggr] \overset{(b)}{\geq} \sum_{t=0}^{T-1}\E\Biggl[C'\Biggl(L_t -\sum_{m,n\in\cB_{t}(m)}\tilde{x}_t\Int(m) \\
&+\sum_{m,n\in\cB_{t+1}(m)}\tilde{x}_{t+1}\Biggr)\Biggl(\sum_{m,n\in\cB_{t}(m)}\tilde{x}_t\Int(m)-  \\
& \sum_{m,n\in\cB_{t+1}(m)}\tilde{x}_{t+1}\Biggr)\Biggr] \overset{(c)}{\geq} 
 \sum_{t=0}^{T-1}\sum_{m=1}^{M}\tilde{x}_t\sum_{n\in\cB_{t}(m)}\E\Biggl[\Int(m)\times \\
& C'\Biggl(L_t-\sum_{j,k\in\cB_t(k)}\tilde{x}_t\I_{k,t}(j)\Biggr)-  C'\Biggl(L_{t-1}+\sum_{j,k\in\cB_{t}(j)}\tilde{x}_t\Biggr)\Biggr] >0.\\
\end{align*}

Inequality (a) follows since Policy A does not necessarily solve \eqref{eq:CP_avg} optimally. Inequality (b) holds by the first order condition on the convexity of the cost function $C$. Inequality (c) follows by rearranging the terms of the RHS of Inequality (b) and replacing the terms $\sum_{j=1}^{M}\sum_{k\in\cB_{t+1}(j)}\tilde{x}_{t+1}$, $-\sum_{j=1}^{M}\sum_{k\in\cB_{t-1}(k)}\tilde{x}_{t-1}\I_{k,t-1}(j)$ with zeros while noting that $C'$ is monotonically increasing on its domain. Finally, the last strict inequality holds since $\tilde{x}_t<\hat{x}_t$ which, combined with the monotonicity and non-negativity of $C'$, yields an always positive sum as long as $\cB_t(m)$ is non-empty for some $t$, $m$.

\section{Proof of Lemma \ref{lem:positive_liminf}}
\label{app:positive_liminf}

The proof follows in two steps. In one step, we show that if $\liminf_{N\to\infty}\tilde{x}_t>0$ then
\begin{multline*}
\liminf_{N\to\infty}\frac{1}{h(N)C'(\gamma_1\cdot N)}\sum_{m=1}^{M}\sum_{n\in\cB_{t}(m)}\E\Biggl[\Int(m)C'\Biggl(L_t- \\
 \sum_{j=1}^{M}\sum_{k\in\cB_{t}(j)}\tilde{x}_t\I_{k,t}(j)\Biggr)\Biggr]>0,
\end{multline*}
 for some $\gamma_1>0$. In the other step, we prove that $\liminf_{N\to\infty}\tilde{x}_t>0$, and $\liminf_{N\to\infty}$
 %\small 
%\begin{align*}
%& \liminf_{N\to\infty}
%\end{align*}
\begin{equation*}
\frac{\underset{m}{\sum}\E\left[\underset{n\in\cB_{t}(m)}{\sum}\Int(m)C'\left(L_t-\overset{M}{\underset{j=1}{\sum}}\underset{k\in\cB_{t}(j)}{\sum}\tilde{x}_t\I_{k,t}(j)\right)\right]}{\underset{m}{\sum}\underset{n\in\cB_{t}(m)}{\sum}\E\left[C'\left(L_{t-1}+\overset{M}{\underset{j=1}{\sum}}\underset{k\in\cB_t(j)}{\sum}\tilde{x}_t\right)\right]}>1.
\end{equation*}
\normalsize
\emph{\textbf{Step 1:}}
Suppose that $\liminf_{N\to\infty}\tilde{x}_t>0$. By Fubini's theorem, we can move the summation inside the expectation, since all the summands are non-negative. Also, by Fatou's lemma, we have
\begin{align*}
\liminf_{N\to\infty}\frac{1}{h(N)C'(\gamma_1\cdot N)} & \\
 \sum_{n\in\cB_{t}(m)}\E\left[\Int(m)C'\left(L_t-\sum_{j=1}^{M}\sum_{k\in\cB_{t}(j)}\tilde{x}_t\I_{k,t}(j)\right)\right] & \geq \\
 \E\Biggl[\liminf_{N\to\infty} \frac{\sum_{m}\sum_{n\in\cB_t(m)}\Int(m)}{h(N)} \cdot & \\
 \end{align*}
 \begin{align*}
  \liminf_{N\to\infty}\frac{C'\left(\sum_{j=1}^{M}\sum_{k=1,k\neq n}^{N}(S(j)-\tilde{x}_t)\I_{k,t}(j)\right)}{C'(\gamma_1 N)}\Biggr] & \overset{(a)}{=}\\
 \liminf_{N\to\infty} \frac{\sum_{m}\sum_{n\in\cB_t(m)}\pnt(m)}{h(N)} \cdot & \\
  \E\Biggl[\liminf_{N\to\infty}\frac{C'\Biggl(N\cdot\sum_{j=1}^{M}\frac{\sum_{k=1,k\neq n}^{N}(S(j)-\tilde{x}_t)\I_{k,t}(j)}{N}\Biggr)}{C'(\gamma_1 N)}\Biggr].
\end{align*}
{Note that, on the left hand side (LHS) of Equality (a), we have removed the contribution of user $n$ in the argument of $C'$ which can only reduce its value, thus yielding $\Int(m)$ independent of the argument of $C'$. Hence, on the RHS, we split the expectation over the product.}
As assumed in the beginning of Section \ref{subsec:asymptotic}, we have $\sum_{m}P_{n,t}(m)\geq\epsilon>0$. Further, for any $n\in\cB_t(m)$, $\pnt(m)>0$ for otherwise $n\notin\cB_t(m)$. Further, $\beta_t(m)>0$ by hypothesis. Therefore, 
\begin{equation*}
\begin{aligned}
\liminf_{N\to\infty}\sum_{m=1}^{M}\sum_{n\in\cB_t(m)}\frac{\pnt(m)}{h(N)}& =\sum_{m=1}^{M}\beta_t(m)>0.
\end{aligned}
\end{equation*}

On the other hand, Kolmogorov's strong law of large numbers implies
\begin{align*}
\gamma_1:&= \lim_{N\to\infty}\sum_{j=1}^{M}\frac{\sum_{k=1,k\neq n}^{N}(S(j)-\tilde{x}_t)\I_{k,t(j)}}{N}\\
         &= \lim_{N\to\infty}\sum_{j=1}^{M}(S(j)-\tilde{x}_t)\frac{\sum_{k\neq n}P_{k,t}(j)}{N} >0 \text{ a.s.},
\end{align*}
since $\qnt<1-\epsilon$, $\forall n,t$. Hence, we have
\begin{align*}
 & \E\Biggl[\liminf_{N\to\infty} \frac{\sum_{m=1}^{M}\sum_{n\in\cB_t(m)}\pnt(m)}{h(N)} \cdot
 \end{align*}
 \begin{align*}
 &  \liminf_{N\to\infty}\frac{C'\Biggl(N\cdot\sum_{j=1}^{M}\frac{\sum_{k=1,k\neq n}^{N}(S(j)-\tilde{x}_t)\I_{k,t}(j)}{N}\Biggr)}{C'(\gamma_1 N)}\Biggr]\\
 & =\sum_{m=1}^{M}\beta_t(m) > 0.
\end{align*}

\emph{\textbf{Step 2:}}
In this step, we prove that $\liminf_{N\to\infty}\tilde{x}_t>0$. To that end, we show that there exists a constant $\chi>0$, independent of $N$, for which if $\tilde{x}_t=\chi$, then
\small
\begin{align*}
& \psi(\chi):=\liminf_{N\to\infty}\\
& \frac{\underset{m}{\sum}\underset{n\in\cB_{t}(m)}{\sum}\E\left[\Int(m)C'\left(L_t-\overset{M}{\underset{j=1}{\sum}}\underset{k\in\cB_{t}(j)}{\sum}\chi\I_{n,t}(j)\right)\right]}{\sum_{m}\underset{n\in\cB_{t}(m)}{\sum}\E\left[C'\left(L_{t-1}+\overset{M}{\underset{j=1}{\sum}}\underset{k\in\cB_t(j)}{\sum}\chi\right)\right]}>1,
\end{align*}
\normalsize
We note that $\psi(\chi)\geq 0$ is continuous and {\em strictly} decreasing in $\chi$.

Now, we have
\begin{align*}
& \psi(0)=\liminf_{N\to\infty} \frac{\underset{m}{\sum}\underset{n\in\cB_{t}(m)}{\sum}\E\left[\Int(m)C'\left(L_t\right)\right]}{\underset{m}{\sum}\underset{n\in\cB_{t}(m)}{\sum}\E\left[C'\left(L_{t-1}\right)\right]}\geq \\
& \frac{\underset{n\in\cup_m \cB_t(m)}{\sum}\delta +\underset{m}{\sum}\E\left[C'\left(L_{t-1}\right)\right]}{\underset{n\in\cup_m\cB_{t}(m)}{\sum}\E\left[C'\left(L_{t-1}\right)\right]}
\end{align*}

By Condition \eqref{eq:delta_condition}, we have $\psi(0)>1$. Thus, by intermediate value theorem (IVT), there exists $\chi>0$ such that $\psi(\chi)>1$.

\section{Proof of Theorem \ref{th:joint}}
\label{app:joint}
We use proof by contradiction as follows. Suppose that $(\vpb,\vxb)$ with $\Intb (m):=1$, with probability $\pntb(m)$ and $0$ with probability $1-\pntb(m)$,
satisfies
\begin{equation}
 \label{eq:hypo1}
 C^{\cJ}(\vpb,\vxb)<C^{\cJ}(\vp^*,\vxs)
 \end{equation}
where $\pntb(m)<1-\qnt, \forall m\in\cM^*$,  {and that user $n$ requests at least one data item under $\vbpnt$ if and only if he requests at least one data item under $\vpnt^*$. That is, the event of requesting a data item is the same in both distributions, only the selected data item can differ depending on the respective distribution.} 
 
Now, {let $\hat{x}_{n,t}(m^*):=\min\{S(m^*),\max_m{\bar{x}_{n,t}(m)}\}$, $\forall n,t$, and consider}  $D:=C^{\cJ}(\vpb,\vxb)-C^{\cJ}(\vp^*,\vxs)$, we have
\begin{align*}
T D = C^{\cJ}(\vpb,\vxb)-\sum_{t=0}^{T-1}\E\Biggl[C\Biggl(\sum_{n=1}^{N}\left(S(m^*)-\xs_{n,t}(m^*)\right)\cdot &\\
\Int^*(m^*) +\xs_{n,t+1}(m^*)\Biggr)\Biggr] \overset{(a)}{\geq}  C^{\cJ}(\vpb,\vxb)- \sum_{t=0}^{T-1}\E\Biggr[& \\
C\Biggl(\sum_{n=1}^{N}\left(S(m^*) -\hat{x}_{n,t}(m^*)\right)\Int^*(m^*)+{\hat{x}}_{n,t+1}(m^*)\Biggr)\Biggr] & \\
\overset{(b)}{\geq} \sum_{t=0}^{T-1}\E\Biggl[C'(Y_t)\cdot \Biggl(\sum_{m,n}\left(S(m)-\bar{x}_{n,t}(m)\right) \Intb(m)+  & \\
\bar{x}_{n,t+1}(m) - \sum_{n=1}^{N}(\left(S(m^*)-\hat{x}_{n,t}(m^*)\right)\Int^*(m^*) & \\
+ \hat{x}_{n,t+1}(m^*))\Biggr)\Biggr]  \overset{(c)}{\geq}  \sum_{t=0}^{T-1}C'\left(\inf Y_t\right)\E\Biggl[\Biggl(\sum_{m,n}(S(m)- & \\
 \bar{x}_{n,t}(m))\Intb(m)+ \bar{x}_{n,t+1}(m) - \sum_{n=1}^{N}\Biggl(\left(S(m^*)-\hat{x}_{n,t}(m^*)\right)\cdot & \\
\Int^*(m^*) +\hat{x}_{n,t+1}(m^*)\Biggr)\Biggr)\Biggr] \text{a.s.} &\\
\end{align*}
where the first equality follows where $\xs_{n,t}(m)=0, \forall m\neq m^*$ as $\pnt^*(m)=0,\forall m\neq m^*$. Inequality (a) holds by replacing $\xs_{n,t}(m^*)$ with $\hat{x}_{n,t}(m^*)$ while noting that $\hat{x}_{n,t}(m^*)$ does not necessarily minimize the expected cost under $\vps$. Inequality (b) holds by the mean value theorem for random variables \cite{mvt} since $Y_t$ is a random variable satisfying
\begin{multline*}
Y_t>\min\Biggl\{\sum_{n=1}^{N}\left(S(m^*)-\hat{x}_{n,t}(m^*)\right)\Int^*(m^*)+\hat{x}_{n,t}(m^*), \\
\sum_{m=1}^{M}\sum_{n=1}^{N}\left(S(m)-\bar{x}_{n,t}(m)\right) \Intb(m)+\bar{x}_{n,t}(m)\Biggr\}\geq 0 \quad \text{a.s.}, 
\end{multline*}
\begin{multline*}
Y_t<\max\Biggl\{\sum_{n=1}^{N}\left(S(m^*)-\hat{x}_{n,t}(m^*)\right)\Int^*(m^*)+\bar{x}_{n,t}(m^*),\\
\sum_{m=1}^{M}\sum_{n=1}^{N}\left(S(m)-\bar{x}_{n,t}(m)\right) \Intb(m)+\bar{x}_{n,t}(m)\Biggr\} \quad \text{a.s.}, 
\end{multline*}
on the entire space of events for all $t=0,\cdots,T-1$. Inequality (c) follows since we consider $\inf$ operator instead of $\E$ operator, while noting that $C'(Y_t)>0, \; a.s.$ as $C$ is an increasing function, {and that}
\begin{multline*}
\sum_{m,n}(S(m)- \bar{x}_{n,t}(m))\Intb(m)+\bar{x}_{n,t+1}(m) \\
 -\sum_{n=1}^{N}\left(\left(S(m^*)-\hat{x}_{n,t}(m^*)\right)\Int^*(m^*)+\hat{x}_{n,t+1}(m^*)\right)\geq 0, \text{a.s.}
\end{multline*}
{by the construction of $\vbpnt$ and $\hat{x}_{n,t}(m^*)$, and the fact that $\sum_{m=1}^{M} {x}_{n,t+1}(m)\geq \hat{x}_{n,t+1}(m^*)$}.

Now, since $C'\left(\inf Y_t\right)>0$ a.s., and

\begin{multline*}
Pr\Biggl(\sum_{m,n}(S(m)- \bar{x}_{n,t}(m))\Intb(m)+\bar{x}_{n,t+1}(m)\\
 -\sum_{n=1}^{N}\left(\left(S(m^*)-\hat{x}_{n,t}(m^*)\right)\Int^*(m^*)+\hat{x}_{n,t+1}(m^*)\right)=0\Biggr)\\
\leq \prod_{n}\qnt<1,
\end{multline*}
it follows that $D> 0$ which contradicts the main hypothesis \eqref{eq:hypo1}.
 
If $|\cM^*|=1$, the uniqueness of $(\vps,\vxs)$ follows since we have proved that $D<0$ for any $m^*\in\cM^*$.

\section{Proof of Theorem \ref{th:gen_approx}}
\label{app:gen_approx}
First, we note that $\fh^{k}$ is convex in $(\vp,\vx)$ since $f_0(\vp^{k-1},\vx)$ is convex in $\vx$ by the definition of the cost function $C$, the term $\sum_{m,n,t}\frac{\partial f_0(\vp,\vx^{k-1})}{\partial \pnt(m)}\biggr |_{\vp=\vp^{k-1}}\cdot\pnt(m)$ is affine in $\vp$, hence convex, and the superposition of convex functions is also convex.

Second, we consider the three conditions specified in Lemma \ref{lem:approx}. Since $f_0$ is continuous in $(\vp,\vx)$ and is defined over a bounded and closed feasible set, it has a global maximum value $U>0$. Such a value can be added to $\fh^k$ defined above to keep Condition 1) of Lemma \ref{lem:approx} satisfied. However, adding a constant to the objective function does not affect the solution. Therefore,  Condition 1) of Lemma \ref{lem:approx} is not necessary in this case.

For Condition 2) of Lemma \ref{lem:approx}, we have
\begin{equation*}
\begin{split}
\frac{\partial \fh(\vp,\vx)}{\partial x_{n,t}(m)}\biggr |_{(\vp^{k-1},\vx^{k-1})}&=\frac{\partial \fh(\vp^{k-1},\vx)}{\partial x_{n,t}(m)}\biggr |_{(\vp^{k-1},\vx^{k-1})}\\
&= \frac{\partial f_0(\vp,\vx)}{\partial x_{n,t}(m)}\biggr |_{(\vp^{k-1},\vx^{k-1})},
\end{split}
\end{equation*}
\begin{equation*}
\frac{\partial \fh(\vp,\vx)}{\partial \pnt(m)}\biggr |_{(\vp^{k-1},\vx^{k-1})}=\frac{\partial f_0(\vp,\vx)}{\partial \pnt(m)}\biggr |_{(\vp^{k-1},\vx^{k-1})},\quad \forall m,n,t.
\end{equation*}
Thus Condition 2) of Lemma \ref{lem:approx} is satisfied.

Finally, Condition 3) of the same lemma need not be satisfied since it is mainly stated in Theorem 1 \cite{Old} for non-convex constraint function that has to be replaced by a convex approximate. Condition 3) mainly implies the satisfaction of the complementary slackness conditions by both the approximate and the original constraint functions. Since we are interested only in the objective function, Condition 3) of Lemma \ref{lem:approx} is not necessary for convergence to a KKT point.

\section{Proof of Theorem \ref{th:DS_gain}}
\label{app:DS_gain}

We construct a suboptimal solution $(\hat{\bf p},\hat{\bf x})$ to \eqref{eq:P_f_horizon}, where $\vhpnt=\vbpntz$, if $(n,t)=(\nz,\tz)$, and    $\vtpnt$ if $(n,t)\neq (\nz,\tz)$, 
and for any $m,n$, $\hxnt(m)=\min\{\txnt(m),S_m-r\}$,  for some $r>0$ if $t=\tz$, and $\txnt(m)$, if  $t\neq \tz$.
Hence, we have $S_m-\hat{x}_{n,\tz}(m)>r$ for all $m,n$.

Now we have $T \Delta C^{\cP}(N) \geq $
\begin{align*}
\sum_{t=0}^{T-1}\E\Biggl[C\Biggl(\sum_{m,n}(S_m-\txnt(m))\tilde{\I}_{n,t}(m)+\tilde{x}_{n,t+1}(m)\Biggr)\\
-C\Biggl(\sum_{m,n}(S_m-\hxnt(m))\hat{\I}_{n,t}(m)+\hat{x}_{n,t+1}(m)\Biggr)\Biggr]  \overset{(a)} {\geq} & \\ 
%\end{align*}
%\begin{align*}
\E\Biggl[C'\Biggl(\sum_{m,n}(S_m-\hxnt(m))\hat{\I}_{n,t}(m)+\hat{x}_{n,t+1}(m)\Biggr)  &  \\
  \sum_{m=1}^{M}\Biggl(S_m-\tilde{x}_{\nz,\tz}(m)\Biggr)(\tilde{\I}_{\nz,\tz}(m)-\hat{\I}_{\nz,\tz}(m))\Biggr] \overset{(b)}{\geq} \\
\E\Biggl[C'\Biggl(\sum_{m,n\neq \nz}(S_m-\hat{x}_{n,\tz}(m))\tilde{\I}_{n,\tz}(m)+\tilde{x}_{n,\tz+1}(m)\Biggr) & \\
 \sum_{m=1}^{M}\Biggl(S_m-\tilde{x}_{\nz,\tz}(m)\Biggr)(\tilde{\I}_{\nz,\tz}(m)-\hat{\I}_{\nz,\tz}(m))\Biggr] \overset{(c)}{\geq} &  \\
  \E\Biggl[C'\Biggl(\sum_{m,n\neq \nz}(S_m-\hat{x}_{n,\tz}(m))\tilde{\I}_{n,\tz}(m)+\tilde{x}_{n,\tz+1}(m)\Biggr) \Biggr] & \\
  \sum_{m=1}^{M} (S_m-\hat{x}_{\nz,\tz}(m))(\tilde{P}_{\nz,\tz}(m)-\bar{P}_{\nz,\tz}(m)).
\end{align*}

Inequality (a) follows by MVT for random variables \cite{mvt} and the construction of $(\vhp,\vhx)$. Inequality (b) follows by removing the contribution of user $\nz$ from the argument of $C'$, thus rendering it independent of $\hat{\I}_{\nz,\tz}$ so we can split the expectation over the product as shown in Inequality (c).

By dividing both sides of the last inequality by $C'(\gamma N)$ for some $\gamma>0$, and sending $N\to\infty$, we have
\begin{align*}
T \liminf_{N\to\infty} \frac{\Delta C^{\cP}(N)}{C'(\gamma N)}  \geq   \sum_{m=1}^{M}(S_m-\hat{x}_{\nz,\tz}(m)) \cdot (\tilde{P}_{\nz,\tz}(m) - & \\
\bar{P}_{\nz,\tz}(m))\cdot  \liminf_{N\to\infty} \frac{1}{C'(\gamma N)}\E\Biggl[C'\Biggl(\sum_{m,n\neq \nz}(S_m -\hat{x}_{n,\tz}(m))\cdot& \\
\tilde{\I}_{n,\tz}(m)+\tilde{x}_{n,\tz+1}(m)\Biggr) \Biggr]=  \sum_{m=1}^{M}(S_m-\hat{x}_{\nz,\tz}(m))\cdot &  \\
 (\tilde{P}_{\nz,\tz}(m)-\bar{P}_{\nz,\tz}(m)) \cdot \liminf_{N\to\infty} \frac{1}{C'(\gamma N)} \E\Biggl[C'\Biggl(N  & \\
\sum_{m,n\neq \nz}\frac{(S_m-\hat{x}_{n,\tz}(m))\tilde{\I}_{n,\tz}(m)}{N}+\tilde{x}_{n,\tz+1}(m)\Biggr) \Biggr].&
\end{align*}

Now, by Fubini's Theorem and Kolmogorov's strong law of large numbers, it follows that
\begin{align*}
\gamma & = \lim_{N\to\infty} \frac{1}{N} \sum_{m=1}^{M}\sum_{n=1,n\neq \nz}^{N}(S_m-\hat{x}_{n,\tz}(m))\tilde{\I}_{n,\tz}(m)\\
       & = \lim_{N\to\infty} \frac{1}{N} \sum_{n=1,n\neq \nz}^{N}\sum_{m=1}^{M} (S_m-\hat{x}_{n,\tz}(m))\tilde{\I}_{n,\tz}(m)>0 \text{ a.s.}\\
\end{align*}
since $\sup_n\{\tilde{q}_{n,\tz}\}<1-\epsilon$. Hence, it is straightforward to see that $ \liminf_{N\to\infty}\frac{\Delta C^{\cP}(N)}{C'(\gamma N)}>0$.


\begin{thebibliography}{1}
\bibitem{cisco1}
Cisco Visual Networking Index: Forecast and Methodology, $2010-2015$.


http://www.cisco.com/en/US/solutions/collateral/ns341/ns525/\\ns537/ns705/ns827/white$\_$paper$\_$c11-481360.pdf

%\bibitem{cisco2}
%Cisco Visual Networking Index: Global Mobile Data Traffic Forecast Update, $2010-2015$.

%http://www.cisco.com/en/US/solutions/collateral/ns341/ns525/\\ns537/ns705/ns827/white$\_$paper$\_$c11-520862.pdf


%\bibitem{FCC2002}
%FCC. Spectrum policy task force report, FCC 02-155. Nov. 2002.
 
%\bibitem{FCC2003}
%FCC. Facilitating opportunities for flexible, efficient, and reliable spectrum
%use employing cognitive radio technologies, notice of proposed rule
%making and order, FCC 03-322. December 2003.

\bibitem{Mitola}
J. Mitola III, ``Cognitive Radio: An Integrated Agent Architecture for Software Defined Radio,''
Doctor of Technology Dissertation, Royal Institute of Technology (KTH), Sweden, May, 2000

%\bibitem{Ian}
%I. Akyildiz, W. Lee, M. Vuran, and S. Mohanty, ``NeXt generation/dynamic spectrum access/cognitive radio wireless networks: A survey,'' \emph{Computer Networks Journal (Elsevier)}, vol. 50, Issue 13, pp. 2127-2159,  September 2006.

\bibitem{Gridlock}
S. A. Jafar, S. Srinivasa, I. Maric, and A. Goldsmith, ``Breaking spectrum gridlock with cognitive radios: an information theoretic perspective,''
\emph{Proceedings of the IEEE}, May 2009.





\bibitem{TEG12}
J. Tadrous, A. Eryilmaz, and H. El Gamal, ``Proactive resource allocation: harnessing the diversity and multicast gains,'' \emph{IEEE Transactions on Information Theory,} vol., no., April 2013.

%\bibitem{Allerton10}
%H. El Gamal, J. Tadrous and A. Eryilmaz, ``Proactive resource allocation: Turning predictable behavior into spectral gain,'' \emph{ 2010 48th Annual Allerton Conference on Communication, Control, and Computing (Allerton),} vol., no., pp.427-434, Sept. 29 2010-Oct. 1 2010.

%\bibitem{ISIT11}
%J. Tadrous, A. Eryilmaz and H. El Gamal, ``Proactive multicasting with predictable demands,'' \emph{2011 IEEE International Symposium on Information Theory Proceedings (ISIT) }, vol., no., pp.239-243, July 31 2011-Aug. 5 2011.

%\bibitem{Asilomar11}
%J. Tadrous, A. Eryilmaz, H. El Gamal, and M. Nafie, ``Proactive resource allocation in cognitive networks,'' \emph{2011 Conference Record of the Forty Fifth Asilomar Conference on Signals, Systems and Computers (ASILOMAR)}, vol., no., pp.1425-1429, 6-9 Nov. 2011.

\bibitem{EP06}
N., Eagle and A. Pentland, ``Reality mining: sensing complex social systems'', \emph{Personal and Ubiquitous Computing}, vol. 10, pp. 255-268, 2006.
\bibitem{FG08}
K. Farrahi and D. Gatica-Perez, ``Discovering human routines from cell phone data with topic models,'' \emph {The 12th IEEE International Symposium on Wearable Computers},  pp. 29-32, 2008.
\bibitem{JLJLH10}
B.S. Jensen, J.E. Larsen, K.   Jensen,  J. Larsen and L. K. Hansen, ``Estimating human predictability from mobile sensor data'', \emph{2010 IEEE International Workshop on Machine Learning for Signal Processing (MLSP)}, pp. 196 - 201, Sept. 2010.
\bibitem{K09}
R. Kwok, ``Personal technology: phoning in data'', \emph{Nature}, vol. 458, pp. 959-961, 2009.
\bibitem{SQBB10}
C. Song, Z. Qu, N. Blumm, and A. Barabas, ``Limits of predictability in human mobility,'' \emph{Science}, vol. 327, pp. 1018-1021, Feb. 2010.
\bibitem{RRDTool}
http://oss.oetiker.ch/rrdtool/gallery/index.en.html

\bibitem{CF1}
Zan Huang, D. Zeng and H. Chen, ``A Comparison of Collaborative-Filtering Recommendation Algorithms for E-commerce,'' \emph{IEEE Intelligent Systems}, vol.22, no.5, pp.68-78, Sept.-Oct. 2007.
%\bibitem{CF2}
%T.L. Wickramarathne, K. Premaratne, M. Kubat and D. T. Jayaweera, ``CoFiDS: A Belief-Theoretic Approach for Automated Collaborative Filtering," \emph{IEEE Transactions on Knowledge and Data Engineering}, vol.23, no.2, pp.175-189, Feb. 2011.
%\bibitem{CF3}
%J. Salter and N. Antonopoulos, ``CinemaScreen recommender agent: combining collaborative and content-based filtering,"  \emph{IEEE Intelligent Systems}, vol.21, no.1, pp. 35- 41, Jan.-Feb. 2006.
%\bibitem{CF4}
%L. Kozma, A. Ilin, and T. Raiko, ``Binary principal component analysis in the Netflix collaborative filtering task,"  \emph{IEEE International Workshop on Machine Learning for Signal Processing, 2009. MLSP 2009.}, vol., no., pp.1-6, 1-4 Sept. 2009.
%\bibitem{CF5}
%O.B. Fikir, I.O. Yaz and T. Ozyer, ``A Movie Rating Prediction Algorithm with Collaborative Filtering," \emph{ 2010 International Conference on Advances in Social Networks Analysis and Mining (ASONAM)}, vol., no., pp.321-325, 9-11 Aug. 2010.
\bibitem{CF6}
G. Adomavicius and A. Tuzhilin, ``Toward the next generation of recommender systems: a survey of the state-of-the-art and possible extensions,'' \emph {IEEE Transactions on Knowledge and Data Engineering,} vol. 17, no.6, pp. 734- 749, June 2005.


\bibitem{BI05}
Y. Bai, and M.R. Ito, ``Proactive resource allocation schemes,'' \emph{IEEE International Conference on Communications, 2005. ICC 2005} , vol.1, no., pp. 53- 58 Vol. 1, 16-20 May 2005.


\bibitem{GH04}
M.J. O'Grady, and G.M.P. O'Hare, ``Just in time multimedia distribution in a mobile computing environment,'' \emph{IEEE MultiMedia}, vol.11, no.4, pp.62,74, Oct.-Dec. 2004.


\bibitem{BWZZ12}
Y. Bao, X. Wang, S. Zhou, and Zhisheng Niu, ``An energy-efficient client pre-caching scheme with wireless multicast for video-on-demand services,'' \emph{2012 18th Asia-Pacific Conference on Communications (APCC)}, vol., no., pp.566,571, 15-17 Oct. 2012.

%\bibitem{GL10}
%V. Glass and P. U. Edge Lab, ``United States broadband goals: managing spillover effects to increase availability, adoption and
%investment,'' white paper, 2010. Available: http://scenic.princeton.edu/paper/NECAPrincetonPaperJune2010.pdf

%\bibitem{E09}
%``The mother of invention: network operators in the poor world are cutting costs and increasing access in innovative ways,'' \emph{The Economist}, Special Report, September 24 2009. 


%\bibitem{H10}
%S. Higginbotham, ``Mobile operators want to charge based on time and apps,'' \emph{gigaOm}, December 14 2010.

%\bibitem{WHSC11}
%S. Ha, C. J.-Wong, S. Sen and M. Chiang, ``Pricing by timing: innovating broadband data plans,'' Proceedings of \emph{SPIE 8282, Broadband Access Communication Technologies}.

%\bibitem{SDP13}
%J. Tadrous, A. Eryilmaz, and H. El Gamal, ``Pricing for demand shaping and proactive download in smart data networks,'' proceedings of \emph{The 2nd IEEE International Workshop on Smart Data Pricing (SDP 2013)}, INFOCOM 2013.

%\bibitem{PB07}
%A. K. Pathan, and R. Buyya, "A taxonomy and survey of content delivery networks," \emph{Tech Report}, University of
%Melbourne, 2007.

%\bibitem{KRR02}
%J. Kangasharju, J. Roberts, and K. Ross, ``Object replication strategies in content distribution networks,'' \emph{Computer Communications}, vol. 25, no. 4, pp. 376–383, March 2002.

%\bibitem{report}
%J. Tadrous, A. Eryilmaz, and H. El Gamal, ``On smart data pricing and proactive downloads,'' \emph{Technical report},
%http://www2.ece.ohio-state.edu/\texttildelow tadrousj/SDPTechReport.pdf

\bibitem{ISIT13}
J. Tadrous, A. Eryilmaz, and H. El Gamal, ``Proactive content distribution for dynamic content'', \emph{2013 IEEE International Symposium on Information Theory (ISIT)}, vol., no., pp.1232,1236, July 2013.

\bibitem {Boyd}
 S. Boyd and L. Vandenberge, ``Convex Optimization,'' \emph{Cambridge University Press}, 2004.
 
\bibitem{mvt}
A. Di Crescenzo,``A probabilistic analogue of the mean value theorem and its applications to reliability theory,'' \emph{Journal of Applied Probability}, vol. 36, no. 3, pp. 706-719, Sept. 1999.
 
 \bibitem{IT}
T. Cover and J. Thomas, ``Elements of information theory,'' \emph{Wiley $-$ Interscience}, 2006.

\bibitem {Old}
B. Marks and G. Wright, ``A general inner approximation algorithm for nonconvex mathematical programs," \emph{Operations Research}, vol. 26, Issue 4, pp. 681-683, 1978.


\bibitem{TR}
J. Tadrous, A. Eryilmaz and H. El Gamal, ``Technical Report: Proactive data download and demand shaping,''  \emph{Technical report},\\ {www2.ece.ohio-state.edu/\texttildelow tadrousj/ProactiveTechReport.pdf}

\end{thebibliography}
\end{document}